\documentclass[a4paper,10pt,reqno]{amsart}
          \usepackage{amssymb}
	  \usepackage{amsmath}
          \usepackage{amsfonts}
          \usepackage[english]{babel}
          \usepackage[utf8]{inputenc}

\usepackage[margin=2.5cm]{geometry}

\usepackage{enumitem}

\usepackage{tikz}   
\usetikzlibrary{arrows,arrows.meta}  
\usepackage{caption}

 \usepackage[unicode,colorlinks,plainpages=false,hyperindex=true,bookmarksnumbered=true,bookmarksopen=false,pdfpagelabels]{hyperref}
 \hypersetup{urlcolor=cyan,linkcolor=blue,citecolor=red,colorlinks=true}
 \usepackage{color}

\newtheorem{thm}{Theorem}[section]
\newtheorem{cor}[thm]{Corollary}
\newtheorem{lem}[thm]{Lemma}
\newtheorem{prop}[thm]{Proposition}

\newtheorem{rem}[thm]{Remark}
\newtheorem{defn}[thm]{Definition}

\numberwithin{equation}{section}

\newcommand{\be}{\begin{equation}}
\newcommand{\ee}{\end{equation}}
\newcommand{\bea}{\begin{eqnarray}}
\newcommand{\eea}{\end{eqnarray}}
\newcommand{\ba}{\begin{aligned}}
\newcommand{\ea}{\end{aligned}}

\begin{document}

\title{Poisson-Lie analogues of spin Sutherland models revisited}

\maketitle

\begin{center}

L. Feh\'er${}^{a,b}$

\medskip
${}^a$Department of Theoretical Physics, University of Szeged\\
Tisza Lajos krt 84-86, H-6720 Szeged, Hungary\\
e-mail: lfeher@sol.cc.u-szeged.hu

\medskip
${}^b$HUN-REN Wigner Research Centre for Physics\\
 H-1525 Budapest, P.O.B.~49, Hungary

\end{center}

\begin{abstract}
Some generalizations of spin Sutherland models
descend from  `master integrable systems'  living on
Heisenberg doubles of compact semisimple Lie groups.
The  master systems represent Poisson--Lie counterparts of the systems of free motion modeled on
the respective cotangent bundles and their reduction  relies on taking
 quotient with respect to  a  suitable conjugation action
of the compact Lie group.
We present an enhanced exposition of the reductions and
prove rigorously for  the first time that the reduced systems possess the property of degenerate integrability
on the dense open subset of the Poisson quotient space  corresponding to the principal orbit type for the
pertinent  group action.
After restriction to a smaller dense open subset,
degenerate integrability  on the generic symplectic leaves is  demonstrated as well.
The paper also contains a novel description of the reduced Poisson structure and
a careful elaboration of the scaling limit whereby our  reduced systems turn into the spin Sutherland models.
\end{abstract}

 \setcounter{tocdepth}{2}

 \tableofcontents

\def\bI{\mathbb{I}}                         %
\def\1{{\boldsymbol 1}}                     %
\def\cD{{\mathcal D}}                       %
\def\cH{{\mathcal H}}                       %
\def\tr{\mathrm{tr}}                        %
\def\diag{\mathrm{diag}}                    %
\def\ri{{\rm i}}                            %
\def\bC{\mathbb{C}}                         %
\def\C{\mathbb{C}}                          %
\def\N{\mathbb{N}}                          %
\def\bR{\mathbb{R}}                         %
\def\T{\mathbb{T}}                          %
\def\cF{{\mathcal F}}                       %
\def\reg{\mathrm{reg}}                      %
\def\red{\mathrm{red}}                      %
\def\Ad{{\mathrm{Ad}}}                      %
\def\id{{\mathrm{id}}}                      %
\def\dt {\left.\frac{d}{dt}\right|_{t=0}}   %
\def\fM{\mathfrak{M}}                       %
\def\cN{{\mathcal N}}                       %
\def\gl{{\rm gl}(n,\C)}                     %
\def\cG{{\mathcal G}}                       %
\def\cR{{\mathcal R}}                       %
\def\I{{\mathbb{I}}}                        %
\def\cB{\mathcal{B}}                       %
\def\Dress{{\mathrm{Dress}}}               %
\def\dress{{\mathrm{dress}}}               %
\def\red{{\mathrm{red}}}                   %
\def\dim{{\mathrm{dim}}}                   %
\def\rank{{\mathrm{rank}}}                 %
\def\cP{\mathcal{P}}                       %
\def\fP{\mathfrak{P}}                      %
\def\ad{\mathrm{ad}}                       %
\def\Ad{\mathrm{Ad}}                       %
\def\cA{\mathcal{A}}                       %
\def\cP{\mathcal{P}}                       %
\def\cU{\mathcal{U}}                       %
\def\cM{\mathcal{M}}                       %
\def\fN{\mathfrak{N}}                      %
\def\End{\mathrm{End}}                     %
\def\r{{\mathrm r}}                        %
\def\fR{\mathfrak{R}}                      %
\def\cC{\mathcal{C}}                       %
\def\fH{\mathfrak{H}}                      %
\def\fC{\mathfrak{C}}                      %
\def\cO{\mathcal{O}}                       %
\def\fF{\mathfrak{F}}                      %
\def\bZ{\mathbb{bZ}}                       %
\def\cL{\mathcal{L}}                       %
\def\bfR{{\mathbf R}}                      %
\def\bM{\mathbb{M}}                        %
\def\cT{\mathcal{T}}                       %
\def\bT{\mathbb{T}}                        %
\def \vLambda{\varLambda}                  %
\newcommand\br[1]{\{ #1 \}}                %
\newcommand\ip[1]{\langle #1 \rangle}      %

\def\u{\mathfrak{u}}                       %
\def\ext{\mathrm{ext}}                     %
\def\lin{{\mathrm{lin}}}                   %
\def\cV{{\mathcal V}}                      %
\def\cW{{\mathcal W}}                      %
\def\o{{\mathrm{o}}}                       %

\newpage

\section{Introduction}
\label{S:1}

It is well known that many important
integrable Hamiltonian systems  can be viewed
as low dimensional `shadows' of higher dimensional manifestly integrable master systems.
The integrability of the master systems is due to their rich symmetries, and their shadows result
by projection onto the quotient space  of the pertinent master phase space with respect to
the symmetry group. This is the essence of the method of Hamiltonian reduction and its variants \cite{OR}.
For example, in the pioneering paper by Kazhdan, Kostant and Sternberg \cite{KKS} the higher dimensional phase space
was the cotangent bundle of the unitary group $\mathrm{U}(n)$, and Marsden--Weinstein reduction at a specific moment map
value for the Hamiltonian action of  $\mathrm{U}(n)$ by conjugations was applied to reduce the master system of free geodesic motion
to the  Sutherland model of $n$ interacting particles on the circle.
Several generalizations of this construction were later investigated in which the group $\mathrm{U}(n)$ was replaced by
other Lie groups or their symmetric spaces \cite{OP,Per}.
It turned out that reductions at generic moment map values often lead to many-body systems possessing
internal `spin' degrees of freedom \cite{FP1,Re1}.
Infinite dimensional master phase spaces built on loop groups \cite{AFM,GN},
and on spaces of flat connections on Riemann surfaces
 \cite{ARe,FR} were also utilized  for constructing integrable systems.
Moreover, there appeared interesting applications of the reduction method \cite{FK0,FK1,FK2,FR,M,Ob}
in which the nature of the underlying symmetry had been
promoted from Hamiltonian group actions to their generalizations based on Poisson--Lie groups \cite{CP,Drin,STS,STSrev}  and on
quasi--Poisson/quasi--Hamiltonian geometry \cite{AKSM,AMM}.
The reviews \cite{A,Eti,N,PolR} and the recent papers \cite{AO,CF,FF,FFM,KLOZ,Re3,Re03} show that this subject possesses
close connections to important areas of physics and mathematics, and enjoys considerable current activity.

The principal goal of the present paper is to complement and enhance our previous results \cite{F1,F2} on the structure
of  Poisson--Lie analogues of those spin Sutherland models that result by reductions of
cotangent bundles of semisimple Lie groups via the conjugation action.
Here, we consider these models in association with every  (connected and simply connected)
compact Lie group $G$ having a simple Lie algebra.
The  relevant master system  is a generalization of the Hamiltonian system on the cotangent bundle
$T^*G$ governed by the kinetic energy of a `free particle' moving on $G$ in the bi-invariant
Riemannian metric.
The master phase space is obtained by replacing the cotangent bundle by the so-called Heisenberg double \cite{STS},
which as a manifold is provided by the complexification $G^\bC$ of the group $G$.
This phase space carries a symplectic structure for which a generalization of the conjugation
action of $G$ on $T^*G$  represents  Poisson--Lie symmetry with respect to the standard multiplicative  Poisson
structure on $G$ \cite{Kli}. There exist also Hamiltonians on the Heisenberg double that generate `free motion' in the sense
that their flows  project on the geodesic lines on $G$.

The free motion modeled on $T^*G$ yields a
degenerate integrable system\footnote{Degenerate integrable systems are also called superintegrable,
the notion as we use it is defined in Section \ref{ss:22}.},
and its reductions by the conjugation action of $G$
 inherit the integrability properties on generic symplectic leaves
of the quotient space $T^*G/G$ \cite{Re1}.
The reduced systems are spin Sutherland models built on `collective spin variables'
belonging to  the reduction of  the dual space $\cG^*$  of the Lie algebra of $G$ (or a coadjoint orbit therein) with respect to the maximal torus  $G_0<G$,
at the zero value of the moment map.
(For the spin Sutherland Hamiltonian, see equation \eqref{Suth}.)
In the Poisson--Lie setting, the space of collective spin variables
becomes a similar reduction of the dual Poisson--Lie group (or a dressing orbit therein), which is the Lie group $G^*=B$
whose Lie algebra $\cB$ appears in the Manin triple \cite{CP,STSrev} displayed in equation \eqref{cGdec}.
The Poisson--Lie analogues of the spin Sutherland models were  first introduced in \cite{F1}, where
 Marsden--Weinstein type reductions
of the Heisenberg double were studied employing the shifting trick of symplectic reduction \cite{OR}.
This means that the phase space was initially extended by a dressing orbit of $G$ in $B$,
and then the reduction was defined by setting the relevant $B$-valued  moment map to the identity value.
The resulting systems were further investigated  in \cite{F2} using Poisson reduction, i.e.,  by directly
taking the quotient of the phase space by the action of the symmetry group $G$.
Via restriction to symplectic leaves after reduction, the two methods give the same models.
The first method is better suited for describing the reduced symplectic form, while the second one
leads more directly to the reduced Poisson algebra.

In  \cite{F1,F2} we collected heuristic arguments in favour of the degenerate integrability
of the reduced systems that descend from the master systems of free motion supported by the Heisenberg doubles,
but have not obtained a full proof.
 The main achievement of this paper is that we will establish in a mathematically
 exact manner the degenerate integrability of the reduced systems after restriction to a dense
 open subset of the Poisson quotient. This subset corresponds to the principal orbit type with respect to the $G$-action on the Heisenberg double.
 After restriction to a smaller dense open subset,
 degenerate integrability on the generic symplectic leaves will be proved as well.
 Our proof of degenerate integrability was motivated by ideas that we learned
 from  papers by Reshetikhin \cite{Re1,Re2}.
 It also relies on techniques introduced in our joint work with
 Fairon \cite{FF} and
 on the recent note \cite{WGMPnew} dealing with reduced integrability on $T^*G/G$.

 Besides the main  achievement,  the  analyses of \cite{F1,F2} will  be  further developed
 in  several other respects as well.
 For example, we shall present two useful, alternative descriptions of the reduced Poisson brackets.
 In the first one `particle positions' and the Lax matrix are used as variables. The explicit formula
 given by equation \eqref{RED1P} contains the dynamical $r$-matrix $\cR(Q)$ \eqref{RQ2} depending on the former.
 The second formulation (given by equation \eqref{FHred-form}) relies on variables
  that may be interpreted as particle positions, their canonical conjugates and
 collective spin degrees of freedom.
 These `decoupled variables'  make it possible to view the reduced systems as Ruijsenaars--Schneider type deformations of the
 standard spin Sutherland models, and we shall present a detailed elaboration of the relevant `scaling limit'.

 \subsection{Organization and results}
 In the next section,
 we  collect the necessary background material concerning Lie theory, degenerate integrability
 and Poisson--Lie symmetry.
 In Section \ref{S:3}, we first give a careful presentation of the Heisenberg double,
 describing its Poisson structure in terms of three distinct sets of variables; each have
 their own advantages as it turns out subsequently. Then, we expose the master system of free motion
 and explain its degenerate integrability.
The core of the paper is Section \ref{S:4}, where we define the reduction of the master system and
demonstrate  the integrability properties of the resulting reduced system.
Our main new results, Theorem \ref{thm:45} with Corollary \ref{cor:46}, and Theorem \ref{thm:49}  can be found in this section.
Section \ref{ss:51} contains the derivation of the dynamical $r$-matrix form
of the reduced Poisson brackets. The result is given by Theorem \ref{thm:defSuthP},
which can be considered as an improvement of a previous result
found in  \cite{F2}.   In Section \ref{ss:52}, we describe the reduced Hamiltonian vector fields  and present
a quadrature leading to their integral curves.
Here, we employ the partial gauge fixing  associated with the gauge slice $\bM_0$ \eqref{bM0reg}, which covers
a dense open subset of the reduced phase space.
Then, in Section \ref{ss:61} we exhibit
canonical conjugates of the position coordinates  and a `collective spin variable' whereby the reduced Poisson bracket
takes the `decoupled form' displayed  in Theorem \ref{thm:5.1}.
This result appeared implicitly already in \cite{F1}, and explicitly in the $G=\mathrm{U}(n)$ case in the paper \cite{FM}.
In Section \ref{ss:62}, we utilize  the decoupled variables
to explain how our reduced systems are connected to the
standard spin Sutherland models in the so-called scaling limit characterized by
Propositions \ref{prop:64} and \ref{prop:65}.  These propositions strengthen and make more precise previous  results of \cite{F1}.
In the final section, we offer a brief summary and an outlook towards
open problems.
There are also three appendices developing technical issues.
 Appendix \ref{sec:A} illustrates how   a Poisson--Lie moment map generates a $G$-action, Appendix \ref{sec:B}  explains a connection
 with the paper \cite{F1}, and in Appendix \ref{sec:C} the previous derivation \cite{FK0}
 of the  spinless trigonometric Ruijsenaars--Schneider model \cite{RS} is recovered from the formalism used
 in the present work.

The exposition of the material that follows is detailed and mostly self-contained,  with the intention
to facilitate further studies of the subject.

\section{Background material}
\label{S:2}

Here, we first summarize a few  Lie theoretic facts for later use.
More details can be found in \cite{F2} and in the textbooks \cite{DK,Knapp,Sam}.
Then, we review the notion of degenerate integrability, and recall
crucial features of Poisson--Lie groups and their actions.

\subsection{Lie theoretic preparations}
 \label{ss:21}

Let $\cG^\bC$ be a complex simple Lie algebra with Killing form  $\langle -, - \rangle$.
The choice of a Cartan subalgebra $\cG_0^\bC < \cG^\bC$ and a system of positive roots leads to the
triangular decomposition
\be
\cG^\bC = \cG_-^\bC + \cG_0^\bC + \cG_+^\bC.
\label{triang}\ee
Then, the `realification' $\cG^\bC_\bR$ of $\cG^\bC$ (i.e. $\cG^\bC$ viewed as a real Lie algebra)
can be  written as the vector space direct sum of two subalgebras
\be
\cG^\bC_\bR = \cG + \cB,
\label{cGdec}\ee
where $\cG$ is a compact simple Lie algebra containing the maximal Abelian subalgebra $\cG_0 < \cG$ for which
\be
\cG_0^\bC = \cG_0 + \ri \cG_0,
\ee
and
\be
\cB := \ri \cG_0 + \cG^\bC_+
\label{cBtriang}\ee
 is a `Borel' subalgebra.
We shall also employ the vector space decompositions
\be
\cG^\bC = \cG_0^\bC + \cG^\bC_\perp
\quad \hbox{with} \quad \cG^\bC_\perp := \cG_-^\bC + \cG_+^\bC,
\label{triangperp}\ee
and
\be
\cG = \cG_0 + \cG_\perp, \quad
\cB = \cB_0 + \cB_+
\label{cGcBperp}\ee
with
\be
\cG_\perp = \cG \cap \cG^\bC_\perp,
\quad
\cB_0 = \ri \cG_0,\quad \cB_+ = \cG^\bC_+.
\ee
Referring to \eqref{triang} and \eqref{triangperp}, we may write any $X\in \cG^\bC_\bR$ as
\be
X = X_- + X_0 +  X_+ = X_0 + X_\perp,
\ee
or by using \eqref{cGdec} as
\be
X= X_\cG + X_\cB,
\label{cGdec+}\ee
and can also further decompose $X_\cG\in \cG$ and $X_\cB\in \cB$ according to \eqref{cGcBperp}.

Let us equip $\cG^\bC_\bR$ with the invariant, nondegenerate, symmetric bilinear form
\be
\langle X, Y\rangle_\bI := \Im \langle X,Y\rangle,
\qquad \forall X,Y\in \cG^\bC.
\label{impair}\ee
  The decomposition  \eqref{cGdec} is a well known example of a Manin triple \cite{CP,STSrev}, meaning that  $\cG$
 and $\cB$ are isotropic subalgebras of $\cG_\bR^\bC$.
 Consequently, the bilinear form gives rise to the following  identifications of linear dual spaces:
 \be
 \cG^* = \cB,  \quad \cB^* = \cG,  \quad
 (\cB_+)^* = \cG_\perp,
 \quad (\cB_0)^* = \cG_0,
 \quad (\cG_0)^* = \cB_0.
\label{dualids}\ee

For the simplest series of examples
$\cG^\bC = \mathrm{sl}(n,\bC)$,  $\cG_0^\bC$ is  the standard Cartan subalgebra of traceless diagonal matrices,
$\cG = \mathrm{su}(n)$, $\cB$ consists of the upper triangular
elements of $\mathrm{sl}(n,\bC)$ with real diagonal entries, and
the Killing form $\langle X, Y\rangle$ is a multiple of $\tr(XY)$ by  a positive constant.

Let $G^\bC_\bR$
be a connected and simply connected real Lie group whose Lie algebra is $\cG^\bC_\bR$, and denote $G$ and $B$ its
connected subgroups associated with the Lie subalgebras $\cG$ and $\cB$.
These subgroups are simply connected and $G$ is compact.
We have the connected subgroup $G_0^\bC <  G^\bC_\bR$ corresponding to $\cG_0^\bC$,
as well as the subgroups $G_0 < G$, $B_0 < B$, $B_+ < B$  associated with $\cG_0$ and the subalgebras $\cB_0$ and $\cB_+$
of $\cB$.

The real vector space $\cG^\bC_\bR$ can be presented as the direct sum
\be
\cG^\bC_\bR = \cG + \ri \cG,
\ee
and we let $\theta$ denote the corresponding  complex conjugation,
\be
\theta( Y_1+ \ri Y_2) := Y_1 - \ri Y_2\ \quad \hbox{for} \quad Y_1,Y_2\in \cG.
\label{theta}\ee
This is an involutive automorphism of the real Lie algebra $\cG_\bR^\bC$, which lifts
to an involutive automorphism $\Theta$ of the group $G^\bC_\bR$.
They are known as infinitesimal and global Cartan involutions, respectively.
It is customary to denote
\be
Z^\dagger:= - \theta(Z), \qquad
K^\dagger= \Theta(K^{-1}),
\qquad\forall Z\in \cG_\bR^\bC,\,\, \forall K\in G_\bR^\bC.
\label{taumap}\ee
The maps $Z \mapsto Z^\dagger$ and $K \mapsto K^\dagger$ are antiautomorphisms.
For the classical Lie groups one can
choose the conventions in such a way that dagger coincides with the matrix adjoint \cite{Knapp}.

The compact subgroup $G < G_\bR^\bC$ is the fixed point set of $\Theta$.
The  closed submanifold
\be
\fP:= \exp(\ri \cG) \subset G^\bC_\bR
\label{fP}\ee
is diffeomorphic to $\ri\cG$ by the exponential map and is
a connected component of the fixed point set of the antiautomorphism $K\mapsto K^\dagger$.
The group $B$ also admits global exponential parametrization, and the map
 \be
 \nu: B \to \fP, \quad \nu(b):= b b^\dagger
\label{nu}\ee
is   a diffeomorphism.

Next, we describe a chain of diffeomorphisms  between the manifolds
\be
M:= G_\bR^\bC, \quad \fM:= G \times B \quad\hbox{and}\quad  \bM:= G \times \fP.
\label{Heis3}\ee
 We start by recalling that
every element $K\in M$ admits unique (Iwasawa) decompositions \cite{Knapp}  into products of elements of $G$ and $B$,
\be
K = g_L b_R^{-1} = b_L g_R^{-1} \quad \hbox{with}\quad g_L, g_R \in G,\, b_L, b_R \in B.
\label{KdecT}\ee
These decompositions induce the (real-analytic) maps $\Xi_L, \Xi_R: M\to G$ and $\Lambda_L, \Lambda_R: M\to B$,
\be
\Xi_L(K) := g_L,\quad \Xi_R(K):= g_R,\quad \Lambda_L(K):= b_L,\quad \Lambda_R(K):= b_R.
\label{XiLaT}
\ee
Besides the pairs $(\Xi_L, \Lambda_R)$ and $(\Xi_R, \Lambda_L)$, also the pair $(\Xi_R, \Lambda_R)$
yields  a diffeomorphism,
\be
m_1:= (\Xi_R, \Lambda_R): M\ \to \fM, \qquad m_1(K) = (g_R, b_R).
\label{m1}\ee
  In addition to this, we need the diffeomorphism
\be
m_2: \fM \to \bM, \qquad
m_2 (g,b): = (g, \nu(b)).
\label{m2}\ee

The map $\nu$ \eqref{nu} intertwines the so-called
 dressing action of $G$ on $B$ with the obvious conjugation action of $G$ on $\fP$.
 That is, we have
\be
\Dress_\eta(b) (\Dress_\eta(b))^\dagger = \eta b b^\dagger \eta^{-1}, \qquad \forall \eta \in G,\, b\in B.
\label{Dressconj}\ee
It follows that any element of $B$ can be transformed into $B_0= \exp(\ri \cG_0)$ by the dressing action.
The relation \eqref{Dressconj} can be taken as the definition of the dressing action. More explicitly, one has
\be
\Dress_\eta(b) = \Lambda_L(\eta b), \qquad
\forall \eta\in G,\, b\in B,
\label{Dress}\ee
and the corresponding infinitesimal action
\be
\dress_Y(b):=
\dt {\mathrm{Dress}}_{e^{tY}}(b)= b (b^{-1} Y b)_\cB, \quad \forall Y\in \cG.
\label{dress}\ee

\begin{rem}
The notation used on the right side of  \eqref{dress}  `pretends' that our Lie groups are groups of matrices.
Such symbolic matrix notations are adopted throughout the paper.
If desired, one may rewrite  all of the relevant equations in equivalent abstract  form (which is often longer), or can employ faithful
matrix representations.
\end{rem}

For a real function $\varphi \in C^\infty(B)$, we define the $\cG$-valued left and right derivatives, $D\varphi$ and $D'\varphi$, by
\be
 \langle X, D \varphi(b) \rangle_\bI + \langle X', D' \varphi(b) \rangle_\bI := \dt \varphi(e^{tX} b e^{tX'}), \quad \forall b\in B, \,
  X,X'
  \in \cB.
 \label{derB}\ee
In general, these obey the relation $D\varphi(b) = ( b D'\varphi(b) b^{-1})_\cG$.
If the function is invariant with respect to the dressing action, $\varphi \in C^\infty(B)^G$,
then $\langle  D'\varphi(b), (b^{-1} Y b)_\cB \rangle_\bI=0$, $\forall Y\in \cG$, and from this we get
$D\varphi(b) = b D'\varphi(b) b^{-1}$.  Equivalently, we have
\be
(b^{-1} D\varphi(b) b)_\cB =0, \qquad \forall \varphi\in C^\infty(B)^G,\, b\in B.
\label{invprop}\ee
By \eqref{dress}, this means that $D\varphi(b)$ belongs to the Lie algebra of the isotropy group $G_b < G$ of $b$
with respect to the dressing action.  Even more, this derivative belongs the center of the isotropy Lie algebra,
because the derivative of an invariant function is equivariant:
\be
D\varphi (\Dress_\eta(b) ) = \eta D\varphi(b) \eta^{-1}, \qquad \forall \eta\in G,\, b\in B.
\ee
The isotropy subgroup $G_b$ is generically a maximal torus of $G$, and the elements for which
this holds constitute the dense open subset $B^\reg \subset B$.
The derivatives of the invariant functions actually span the center of the isotropy Lie algebra at any $b\in B$.
This can be seen, for example, with the aid of the natural isomorphisms
\be
C^\infty(\cG_0)^W \longleftrightarrow C^\infty(\cG)^G \longleftrightarrow C^\infty(\fP)^G \longleftrightarrow
 C^\infty(B)^G,
\ee
where $C^\infty(\cG_0)^W$ denotes the Weyl invariant smooth functions on $\cG_0$.
The isomorphisms are induced by the  maps
 \be
 \cG_0 \stackrel{\iota}{\longrightarrow}\cG \xrightarrow{\exp_\ri}   \fP \stackrel{\nu}{\longleftarrow} B,
 \label{arrowtop}\ee
 where $\iota:\cG_0 \to\cG$ is the  inclusion, $\exp_\ri(X):= \exp(\ri X)$, and $\nu$ is defined in \eqref{nu}.
  These maps also
 relate the dense open subsets $\cG_0^\reg$, $\cG^\reg$, $\fP^\reg$ and $B^\reg$.
 It is  well known that the gradients (with respect to the Killing form of $\cG)$
 of the invariant functions on $\cG$ span
 the center of the corresponding isotropy subalgebra. The dimension of the span
 of the derivatives of the invariant functions does not change under these maps,
 since the derivatives are equivalent to the ordinary exterior derivatives
 (for example, $D\varphi (b)\in \cG \simeq \cB^*$ encodes
 $d \varphi(b)\in  T_b^*B$).

\subsection{Degenerate integrability}
 \label{ss:22}

The notion of degenerate integrability of Hamiltonian systems on symplectic manifolds is due to
Nekhoroshev \cite{Nekh}. Degenerate integrable systems have more first integrals (constants of motion)
than half the dimension
of the phase space, which characterizes Liouville integrability.
In the extreme case the trajectories are completely determined by fixing the constants of motion,
as is exemplified by the classical Kepler problem that possesses 5 independent constants of motion.
A closely related concept of non-Abelian or non-commutative integrability was introduced
by  Mischenko and Fomenko \cite{MF}, and  this is especially fitting for
systems whose basic constants of motion form a finite-dimensional, non-Abelian Lie algebra.

A systematic exploration of natural quantum mechanical Hamiltonians with
many conserved quantities having specific form was initiated by Fris \emph{et al } in 1965 \cite{Fris}, and later this has become
a very active research subject \cite{MPV}.
The systems studied in this field are nowadays called   superintegrable,
a term that apparently goes back to Wojciechowski \cite{Wo}.
The adjective \emph{superintegrable } is now often used to characterize both quantum and classical
mechanical systems \cite{Fas,MPV,Re3,Re03}.
We prefer to stick to the original terminology of Nekhoroshev, which  highlights the important feature
that in comparison to  classical Liouville integrability
the dynamics takes place on  lower dimensional submanifolds of the phase space
(typically, degenerate tori when compact).
The extensive literature on the subject of
integrability (see e.g.~\cite{J,LMV,Rud}) contains
several variants of the basic notions.
The definition  that we find the most convenient is presented below.

\begin{defn}\label{def:21}
Suppose that $\cM$ is a symplectic manifold of dimension $2m$ with associated Poisson bracket
$\br{-,-}$ and two distinguished subrings $\fH$ and $\fF$ of $C^\infty(\cM)$
satisfying the  following conditions:
\begin{enumerate}[itemsep=0pt]
\item{
The ring $\fH$ has functional dimension $r$ and $\fF$ has functional dimension $s$ such that
$r + s = \dim(\cM)$ and  $r<m$.}
\item{Both $\fH$ and $\fF$ form Poisson subalgebras of $C^\infty(\cM)$,  satisfying
$\fH\subset \fF$  and $ \{\cF, \cH\}=0$ for all $\cF\in \fF$, $\cH\in \fH$.}
\item{The Hamiltonian vector fields of the elements of $\fH$ are complete.}
\end{enumerate}
Then, $(\cM, \br{-,-}, \fH, \fF)$ is called a  degenerate integrable system of rank $r$.
The rings $\fH$ and $\fF$ are referred to as the ring of Hamiltonians and constants of motion, respectively.
\end{defn}

Recall that the functional dimension of a ring $\fR$ of functions
on a manifold $\cM$ is $d$
if the exterior derivatives of the elements of $\fR$  \emph{generically}, that is on a dense open submanifold,
span a $d$-dimensional subspace of the cotangent space.
Condition $(3)$ above on the completeness of the flows is superfluous if the joint level surfaces
of the elements of $\fF$  are compact.
Degenerate integrability of a single Hamiltonian $\cH$ is  understood to mean
that there exist rings $\fH$ and $\fF$ with the above properties such that $\cH\in \fH$.
Observe that $\fF$ is either equal to or can
be enlarged to the centralizer of $\fH$ in the Poisson algebra $(C^\infty(\cM), \br{-,-})$.
In the literature the definition is often formulated in terms of functions
$f_1,\dots, f_r, f_{r+1},\dots, f_s$ so that they generate $\fF$ and the first $r$ of them generate $\fH$.
If the definition is modified by setting $r=s=m$ and $\fH = \fF$, then one obtains
the notion of Liouville integrability.

The concepts of integrability can be extended to Poisson manifolds
\cite{LMV} beyond the symplectic class.
In fact, we shall construct a series of examples that satisfy the requirements of the
next definition.

\begin{defn}\label{def:22}
Consider
 a Poisson manifold $(\cM,\br{-, -})$  whose Poisson tensor has maximal rank $2m\leq \dim(\cM)$
on a dense open subset.  Then,  $(\cM,\br{-,-},\fH,\fF)$ is called a degenerate integrable
system of rank $r$ if conditions (1), (2), (3) of Definition \ref{def:21}  hold, and
the Hamiltonian vector fields of the elements of $\fH$ span an $r$-dimensional subspace of
the tangent space over a dense open subset of $\cM$.
\end{defn}

The integrable systems of Definition \ref{def:21}  are integrable in the sense
of Definition \ref{def:22}, too, since in the symplectic case the condition on the span of
the Hamiltonian vector fields of $\fH$ holds automatically.
Liouville integrability in the Poisson case results by imposing $r=m$
instead of $r<m$ in the definition.
In that case, our definition implies that $\fF$ is an Abelian Poisson algebra
(in \cite{LMV} this condition appears in the definition).

In a degenerate integrable system,
the evolution equation associated with  any $\cH\in \fH$
can be integrated by quadrature (see, e.g., \cite{J,Nekh,Rud}).
A description of `action-angle and spectator' coordinates in the Poisson case
can be found in \cite{LMV}.
Under further conditions, it can be shown \cite{J} that
degenerate integrable systems are integrable also in the Liouville sense.
However, in general
there is no canonical way to enlarge $\fH$ by elements of $\fF$ to obtain
an Abelian Poisson algebra of the required functional dimension.
This freedom can be used to manufacture very different  Liouville integrable systems
out of a given degenerate integrable system.
For spin Calogero--Moser type systems, and their generalizations that we are interested in,
 $\fH$  is distinguished by its group theoretic origin \cite{F2,Re1,Re2}.

\subsection{Poisson--Lie symmetry}
 \label{ss:23}

Poisson--Lie groups are the quasi-classical analogues of quantum groups introduced by Drinfeld \cite{CP,Drin}.
Their role in classical integrable systems was pioneered by Semenov--Tian--Shansky \cite{STS}, whose
review \cite{STSrev} is highly recommended as a general reference.

By definition, a Poisson--Lie group is a pair $(G,\br{-,-}_G)$, where
$\br{-,-}_G$ is a Poisson bracket on the smooth (or holomorphic etc)
functions on the Lie group $G$  such that the group product
$G\times G \to G$ is a Poisson map.
A Poisson action of $(G,\br{-,-}_G)$ on a Poisson manifold $(\cM,\br{-,-})$ is an action for which the action map
$\cA: G \times \cM \to \cM$ is Poisson.
In these definitions, $G\times G$ and $G \times \cM$ are equipped with the respective product Poisson structures.
Take arbitrary points $g\in G$, $p\in \cM$ and for any $F\in C^\infty(\cM)$ define $F_g\in C^\infty(\cM)$ and
$F^p \in C^\infty(G)$ by
\be
F_g(p) = F^p(g)= F(\cA_g(p))
\ee
using $\cA_g(p):= \cA(g,p)$. The Poisson property of the map $\cA$ means that
\be
\{ F,H\}(\cA_g(p)) = \{ F_g, H_g\}(p) + \{ F^p, H^p\}_G(g), \qquad \forall F,H \in C^\infty(\cM).
\label{Poact}\ee

The Poisson tensor of every Poisson Lie group vanishes at the  unit element $e\in G$. Thus, the linearization of the Poisson bracket $\br{-,-}_G$
 yields \cite{CP,STSrev} a Lie bracket $[-,-]_*$ on the dual space $\cG^*= T^*_e G$ of the
Lie algebra $\cG= T_eG$.
For any $X\in \cG$, let $X_\cM$ denote the infinitesimal generator of the (left) $G$-action on $\cM$, such
that $X \mapsto X_\cM$ is an antihomomorphism, and let $\cL_{X_\cM}F= dF(X_\cM)$ denote the derivative of the  function
$F\in  C^\infty(\cM)$.
Pick a  basis $\{T_a\}$ of $\cG$ with dual basis $\{T^a\}$ of $\cG^*$, and define
$\zeta_F \in C^\infty(\cM, \cG^*)$ by
\be
\zeta_F:= \sum_a T^a  dF( (T_a)_\cM),\qquad \forall F\in C^\infty(\cM).
\label{zetaF}\ee
Then, the Poisson property \eqref{Poact}  implies  the identity
\be
\cL_{X_\cM} \{F, H\} -  \{\cL_{X_\cM} F, H\} -  \{F, \cL_{X_\cM}  H\}  -  ([\zeta_F, \zeta_H]_*, X)  =0,
\label{PLid1}\ee
for all $X\in \cG$, $F,H \in C^\infty(\cM)$, where in the last term  the pairing between $\cG^*$ and $\cG$ is used.
Indeed, \eqref{PLid1} follows by putting $g = \exp(tX)$ in \eqref{Poact} and taking derivative with respect to $t\in \bR$ at $t=0$.
It is also worth noting that
\be
\zeta_F(p) = (d_G F^p)(e) \quad \hbox{and}\quad  [\zeta_F(p), \zeta_H(p)]_* = (d_G \{ F^p, H^p\}_G)(e),
\ee
where $d_G$ denotes the exterior derivation of functions on $G$.

We assume that $G$ is connected, and then \eqref{PLid1} is  equivalent to the Poisson property of the $G$-action.
Two consequences of the identity \eqref{PLid1} are important for us.
First, if both $F$ and $H$ are $G$-invariant, then so is their Poisson bracket, i.e.,
$C^\infty(\cM)^G$ is closed under the Poisson bracket.
The statement holds because
$\cL_{X_\cM} \{F, H\} =0$ in this case.
Second, if $F\in C^\infty(\cM)$ is arbitrary and $H\in C^\infty(\cM)^G$, then \eqref{PLid1} becomes
\be
\cL_{X_\cM} \{F, H\} -  \{ \cL_{X_\cM}  F,H\}  =0.
\ee
Defining the Hamiltonian vector field $V_H$ by $\{F,H\} =: \cL_{V_H}(F)$, the identity means
that
\be
[X_\cM, V_H] =0,
\qquad \forall X\in \cG,\, H\in C^\infty(\cM)^G.
\label{PLid2}\ee
This entails that the corresponding flows, denoted $\varphi_\tau^X$ and $\varphi^H_t$, commute
\be
\varphi_\tau^X  \circ\varphi^H_t = \varphi^H_t \circ \varphi^X_\tau.
\label{PLid3}\ee
Since $G$ is supposed to be connected, this in turn implies that the Hamiltonian flow $\varphi_t^H$ is $G$-equivariant.
In favourable circumstances, e.g. if the group $G$ is compact, one may identify
$C^\infty(\cM)^G$ with the ring of smooth functions on the quotient space $\cM^\red:= \cM/G$.
In this way, $C^\infty(\cM^\red)$ becomes  a Poisson algebra, and the Hamiltonian flows generated by its elements are
the projections of the flows $\varphi_t^H$ living upstairs.
The process of descending to the quotient space $\cM^\red$ is known as Poisson reduction,
or  Hamiltonian reduction if a $G$-invariant Hamiltonian is also specified.
It should be noted that the quotient space $\cM^\red$ is usually not a smooth Poisson manifold, but a so-called
stratified Poisson space \cite{OR,SL}.

In the theory of Poisson actions of $(G,\br{-,-}_G)$
the $G^*$-valued Poisson--Lie moment map plays an important role \cite{LuPhD,Lu}.
Here, $G^*$ is the dual Poisson--Lie group  \cite{CP,LuPhD,STSrev}, whose Lie algebra is $(\cG^*, [-,-]_*)$ mentioned above and
the linearization  of the Poisson bracket on $G^*$  reproduces the Lie algebra of $G$.
The precise notion of the moment map  will be recalled in Section \ref{S:4} focusing on the  groups our  interest.
The Poisson--Lie moment map can be used for finding Poisson subspaces  of $\cM^\red$
quite in the same way as for the standard $\cG^*$-valued moment map \cite{OR}.  For compact semisimple Lie groups,
there is a direct link between ordinary Hamiltonian $G$-actions and their Poisson--Lie analogues.
One can be converted into the other by means of a modification of the symplectic form, without changing the reduced structure \cite{Al}.

\section{Integrable master system on the Heisenberg double}
\label{S:3}

In the first subsection we give a terse overview of the Poisson geometry
of the standard Heisenberg double of the compact Lie group $G$.
The second subsection is  devoted to the description
of a degenerate integrable system on this phase space.

\subsection{Three models of the Heisenberg double}
 \label{ss:31}

We recall \cite{STS,STSrev} that the group manifold
$M= G^\bC_\bR$
 carries the following two  Poisson brackets:
\be
\{ \Phi_1, \Phi_2\}_{\pm}: = \langle \nabla \Phi_1, \rho \nabla \Phi_2 \rangle_\bI \pm  \langle \nabla' \Phi_1, \rho \nabla' \Phi_2 \rangle_\bI,
\quad \forall \Phi_1, \Phi_2 \in C^\infty(M).
\label{A1T}\ee
Here, $\rho := \frac{1}{2}\left( \pi_{\cG} - \pi_{\cB}\right)$ with
 $\pi_\cG$ and $\pi_\cB$  denoting the projections from $\cG^\bC_\bR$ onto $\cG$ and $\cB$, which correspond to the direct sum in \eqref{cGdec}.
For any real function $\Phi \in C^\infty(M)$,
the $\cG^\bC_\bR$-valued  left  and right derivatives  $\nabla \Phi$ and $\nabla' \Phi$ are defined by
 \be
 \langle X, \nabla \Phi(K) \rangle_\bI + \langle X', \nabla' \Phi(K) \rangle_\bI := \dt \Phi(e^{tX} K e^{tX'}),
 \quad \forall K\in M, \, X,X' \in \cG_\bR^\bC.
 \label{Nab}\ee
 The minus bracket makes $M$ into a Poisson--Lie group, of which $G$ and $B$ are Poisson--Lie subgroups, i.e.,
 (embedded) Lie subgroups and Poisson submanifolds.
 Their \emph{inherited Poisson brackets} take the form
\be
\{ \chi_1, \chi_2\}_G(g) = - \langle D' \chi_1(g), g^{-1} (D \chi_2(g)) g \rangle_\bI,
\label{PBGT}\ee
and
\be
\{ \varphi_1, \varphi_2\}_B(b) = \langle D' \varphi_1(b), b^{-1} (D \varphi_2(b)) b \rangle_\bI.
\label{PBBT}\ee
The derivatives are $\cB$-valued for $\chi_i\in C^\infty(G)$ and $\cG$-valued for  $\varphi_i\in C^\infty(B)$.
Concretely, we use the definition \eqref{derB} for any $\varphi \in C^\infty(B)$, and
\be
 \langle Y, D \chi(g) \rangle_\bI + \langle Y', D' \chi(g) \rangle_\bI := \dt \chi(e^{tY} g e^{tY'}), \quad \forall g\in G, \, Y,Y' \in \cG,
\label{derG} \ee
for any $\chi\in C^\infty(G)$.
The Poisson manifolds $(M, \br{- ,- }_-)$
and $(M, \br{- ,- }_+)$ are known, respectively, as the Drinfeld double and the Heisenberg double
associated with the standard Poisson structures of $B$ and $G$.
The Poisson bracket $\br{-,- }_+$  is nondegenerate.
The corresponding symplectic form was found in \cite{AM}, but we shall not use its formula here.
It is also known that the maps
\be
(\Lambda_L, \Lambda_R): M \to B \times B
\quad\hbox{and}\quad
(\Xi_L, \Xi_R): M \to G \times G
\label{LRmaps}\ee
are Poisson maps with respect to $(M, \br{-,-}_+)$ and the direct product
Poisson structures on the targets obtained from $(B,\br{-,-}_B)$ and from $(G,\br{-,-}_G)$, respectively.

Below we focus on the Heisenberg double $(M, \br{- ,- }_+)$, and transfer its Poisson structure to $\fM$ and $\bM$ \eqref{Heis3}
 by the diffeomorphisms
$m_1$ \eqref{m1} and $m_2$ \eqref{m2}.
As was proved in \cite{F2}, the usage of $m_1$ results in the Poisson bracket $\br{-,-}_\fM$ on $\fM$ having
 the following explicit form:
\be
\{f, h\}_\fM(g,b) =\left\langle D_2' f, b^{-1} (D_2 h) b \right\rangle_\bI
-\left\langle D'_1 f, g^{-1} (D_1 h) g\right\rangle_\bI
 +  \left\langle D_1 f , D_2 h \right\rangle_\bI
-\left\langle D_1 h , D_2 f \right\rangle_\bI
\label{A2T}\ee
for functions $f, h\in C^\infty(\fM)$.
The
 derivatives on the right-hand side
are   taken at $(g,b)\in G\times B$, with respect to the first and second variable,
according to the definitions \eqref{derG} and \eqref{derB}, respectively.
In particular, $D_1 f$ is $\cB$-valued and $D_2 f$ is $\cG$-valued.

For any real function $\cF\in C^\infty(\fP)$, define its \emph{$\cG^\bC_\bR$-valued derivative}  $\cD \cF$ as follows:
 \be
 \langle X, \cD \cF(L)\rangle_\bI :=  \dt \cF( e^{t X} L e^{t X^\dagger}), \qquad \forall X\in \cB,
 \label{newD1}\ee
 with $X^\dagger$ given by \eqref{taumap},  and
 \be
 \langle Y, \cD \cF(L)\rangle_\bI :=  \dt \cF( e^{t Y} L e^{-t Y}), \qquad \forall Y\in \cG.
 \label{newD2}\ee
 Referring to \eqref{cGdec+}, the first equation determines $(\cD \cF(L))_\cG$ and the second one $(\cD \cF(L))_\cB$.
 (These two equations could be `unified' since $- Y = Y^\dagger$ for $Y\in \cG$, but we prefer to display them separately.)
 Because the natural action of $B$ on $\fP$ is
 transitive\footnote{The corresponding action map is
  $B\times \fP\in (b, L)\mapsto b L b^\dagger\in \fP$.},
  all information about $ \cD \cF$ is  contained in the $\cG$-component.
 This is clear from the next lemma, too.

 \begin{lem}  \label{LemDerP}
 Let $\cF\in C^\infty(\fP)$ and $\varphi\in C^\infty(B)$ connected by the diffeomeorphism $\nu$ \eqref{nu}, i.e.,
 \be
 \cF(bb^\dagger) = \varphi(b), \qquad \forall b\in B.
\ee
Then their derivatives satisfy
\be
(\cD \cF(bb^\dagger))_\cG = D \varphi(b) \equiv (b D' \varphi(b) b^{-1})_\cG,
\qquad
(\cD \cF(bb^\dagger))_\cB = (b D' \varphi(b) b^{-1})_\cB,
\label{newDrel1}\ee
and consequently
\be
\cD \cF(bb^\dagger) = b D'\varphi(b) b^{-1}.
\label{newDrel2}\ee
\end{lem}

\begin{proof}
Take $L= b b^\dagger$ and consider the curve $ e^{tX} $ for $X\in \cB$. Since $\nu( e^{t X} b)  =  e^{t X} L e^{t X^\dagger}$,
 we obtain
 \be
  \langle X, b D'\varphi(b) b^{-1} \rangle_\bI = \langle X, D\varphi(b) \rangle_\bI =  \dt \varphi( e^{t X} b) =\dt \cF (e^{tX} L e^{t X^\dagger})
  = \langle X, \cD \cF(L) \rangle_\bI,
  \ee
which implies the first equality in \eqref{newDrel1}.
Next, take any $Y\in \cG$ and consider the curve ${\mathrm{Dress}}_{e^{tY}}(b)$.
Using \eqref{dress} we get
\be
\dt \varphi( {\mathrm{Dress}}_{e^{tY}}(b)) = \langle D'\varphi(b), (b^{-1} Y b)_\cB \rangle_\bI = \langle b D'\varphi(b) b^{-1}, Y \rangle_\bI.
\ee
Due to the identity
$\nu( {\mathrm{Dress}}_{e^{tY}}(b)) = e^{tY} L e^{-tY}$, this is equal to
\be
\dt \cF(e^{tY} L e^{-tY}) = \langle (\cD \cF(L))_\cB, Y\rangle_\bI.
\ee
Consequently,  the  second equality in \eqref{newDrel1} holds, too.
\end{proof}

We use the map $\nu$ \eqref{nu} to transfer the Poisson bracket $\br{-,-}_B$ \eqref{PBBT} from $B$ to $\fP$.
 With the aid of Lemma \ref{LemDerP}, this leads to
 \be
 \{\cF, \cH\}_\fP(L) = \left\langle (\cD \cF(L))_\cB, (\cD \cH(L))_\cG \right\rangle_\bI= - \left\langle (\cD \cF(L))_\cG, (\cD \cH(L))_\cB \right\rangle_\bI,
 \quad \forall \cF,\cH\in C^\infty(\fP).
 \label{PBonP}\ee
 The Hamiltonian vector field generated by $\cH\in C^\infty(\fP)$ on $\fP$ yields the evolution equation
 \be
 \dot{L} = [ (\cD\cH(L))_\cG, L].
 \label{dotL}\ee
 The identical vanishing of the right-hand side characterizes the elements of the center of the Poisson bracket \eqref{PBonP}.
The elements of the center are constant along all Hamiltonian flows, and this gives their equivalent characterization:
 $(\cD \cH(L))_\cB =0$. This holds if and only $\cH$ is $G$-invariant, that is, for $\cH\in C^\infty(\fP)^G$.
As one can verify,  the derivative of every invariant function is equivariant,
 \be
 \cD \cH(\eta L \eta^{-1}) = \eta \cD\cH(L) \eta^{-1}, \qquad \forall \eta\in G,\, L\in \fP,\, \cH\in C^\infty(\fP)^G.
 \ee

 Similarly to the relation between $B$ and $\fP$, we can transfer the Poisson bracket \eqref{A2T}
 from $\fM$ to $\bM$ \eqref{Heis3} via the map $m_2$ \eqref{m2}.
 To display the result, for $\cF \in C^\infty(\bM)$ let $\cD_1\cF(g,L) \in \cB$ and  $\cD_1'\cF(g,L) \in \cB$ denote the
 usual derivatives with respect to the first variable, and $\cD_2\cF(g,L)\in \cG^\bC_\bR$ denote the derivative
 with respect to the second variable defined according to equations \eqref{newD1}  and \eqref{newD2}.

 \begin{prop}  Via the map $m_2$ \eqref{m2},
 the formula \eqref{A2T} of the Poisson bracket on $\fM$ is equivalent to
 \be
 \{ \cF, \cH\}_\bM(g,L)=\left\langle \cD_2 \cF, (\cD_2\cH)_\cG \right\rangle_\bI
-\left\langle g \cD'_1\cF g^{-1},  \cD_1\cH \right\rangle_\bI
 +  \left\langle \cD_1\cF , \cD_2\cH \right\rangle_\bI
-\left\langle \cD_1 \cH , \cD_2\cF \right\rangle_\bI,
\label{A3T}\ee
where the derivatives of $\cF, \cH\in C^\infty(\bM)$ are evaluated at $(g,L)\in \bM=G\times \fP$.
 \end{prop}
 \begin{proof}
 We simply substitute the following relations into \eqref{A2T}:
 \be
 D_1f(g,b) = \cD_1 \cF(g,L),
 \quad
 b D_2'f(g,b) b^{-1} = \cD_2\cF(g,L),
 \quad
 D_2 f(g,b)= (\cD_2 \cF(g,L))_\cG.
 \ee
  Using also the corresponding relations for $h = \cH \circ m_2$,  we  get \eqref{A3T} from \eqref{A2T}.
 Note that $\langle \cD_1 \cF, \cD_2 \cH\rangle_\bI = \langle \cD_1 \cF, (\cD_2 \cH)_\cG\rangle_\bI$,
 because $\cD_1 \cF(g,L)\in \cB$.
 \end{proof}

 \begin{rem}\label{rem:Luconv}
For the alert reader, a word on `tricky signs' is in order.
 Using the identifications $\cB^* = \cG$ and $\cG^* = \cB$ in \eqref{dualids},
 the linearization of the Poisson bracket \eqref{PBBT} on $B$ gives
 the Lie bracket on the subalgebra $\cG < \cG^\bC_\bR$, and the linearization of the opposite of the Poisson bracket
  on $G$ gives the Lie bracket on  $\cB < \cG^\bC_\bR$. In standard terminology, this means that
  $(G, (-1)\br{-,-}_G)$ and $(B,\br{-,-}_B)$ form a pair of mutually dual Poisson Lie groups.
The dual group of $(G, \br{-,-}_G$) is obtained from $(B, \br{-,-}_B)$ by keeping the Poisson
structure on the manifold $B$  but replacing the group product by its opposite, defined
by $b_1 \star b_2 = b_2 b_1$, which also changes the corresponding Lie bracket on $\cB$ to its opposite.
 \end{rem}

\subsection{The integrable master system of free motion}
 \label{ss:32}

By the projection
\be
\pi_2: \bM \to \fP, \quad \pi_2(g,L) = L,
\label{pi2}\ee
we can pull-back the elements of $C^\infty(\fP)^G$ to $\bM$ \eqref{Heis3}. Since $\pi_2$ is a Poisson map, this yields
Poisson commuting Hamiltonians on $\bM$. We next describe the flows and the constants of motion
for these Hamiltonians.

\begin{prop}\label{prop:Pint}
Let $\cH = \pi_2^*(\phi)$ for a function $\phi\in C^\infty(\fP)^G$ and pick an initial value
$(g(0), L(0))\in \bM$.
The corresponding integral curve of the  Hamiltonian vector field of $\cH$, defined by means of $\br{- ,- }_\bM$ \eqref{A3T},
is provided by
\be
(g(t), L(t)) = \left(\exp\left( t \cD \phi(L(0)) \right) g(0), L(0) \right).
\label{Pint}\ee
The map $\Psi: \bM \to \fP \times \fP$  defined by
\be
\Psi(g,L) := (\tilde L, L) \quad \hbox{with}\quad \tilde L:= g^{-1} L g
\label{Psi}\ee
is constant along the integral curves \eqref{Pint}.
The map $\Psi$ is Poisson  with  respect to \eqref{A3T} and the direct product Poisson structure on $\fP \times \fP$
obtained from $\br{-,-}_\fP$ \eqref{PBonP} on the second $\fP$ factor and its opposite
(multiple by $-1$)
on the first $\fP$ factor.\footnote{When we wish to emphasize its Poisson structure, we denote this space as $\fP_- \times \fP$.}
\end{prop}
\begin{proof}
For the Hamiltonian $\cH$ at hand, $\cD_1\cH(g,L)=0$ and $\cD_2\cH(g,L) = \cD \phi(L) \in \cG$.
Therefore we have
\be
\{\cF, \cH\}_\bM(g,L) = \langle \cD_1\cF(g,L), \cD\phi(L)\rangle_\bI, \quad \forall \cF\in C^\infty(\bM).
\ee
Hence, $\cH$ generates the evolution equation
\be
\dot{g} = (\cD\phi(L)) g, \qquad \dot{L}=0,
\label{dotgL}\ee
which is  solved by \eqref{Pint}. The statement that $\tilde L$ is constant along the flows is
then verified by using that $[L, \cD \phi(L)]=0$ for all $\phi \in C^\infty(\fP)^G$.

To see the Poisson map property of $\Psi$, consider the diffeomorphism
\be
m:= m_2 \circ m_1: M \to \bM,
\label{m}\ee
based on \eqref{m1} and \eqref{m2}.
From \eqref{KdecT},  we have
\be
g_R^{-1} b_R = b_L^{-1} g_L, \qquad \forall K\in M.
\label{LRrel}\ee
Using also \eqref{nu}, this implies the equality
\be
\Psi \circ m = \left(  \nu\circ (\Lambda_L)^{-1}, \nu \circ \Lambda_R\right ).
\ee
The Poisson property of $\Psi$ is then a consequence of the facts that $(\Lambda_L, \Lambda_R): M \to B \times B$
is a Poisson map, where $B\times B$ carries the product Poisson structure with $\br{-,-}_B$
on the two copies, and that taking the inverse is an anti-Poisson map on every Poisson--Lie group.
\end{proof}

It is worth noting that the map $\Psi$ is not surjective and its image
\be
\fC:= \Psi(\bM) \subset \fP \times \fP
\label{fC}\ee
is \emph{not} a smooth manifold in any natural way. However,  the dense subset $\fC_\reg \subset \fC$, given by
\be
\fC_\reg:= \Psi(\pi_2^{-1}(\fP^\reg)),  \qquad
 \pi_2^{-1}(\fP^\reg) = G \times \fP^\reg,
\label{fCreg1}\ee
is an embedded submanifold of $\fP^\reg \times \fP^\reg$ of co-dimension $r$.
(Recall that $\fP^\reg$ contains those elements of $\fP$ whose isotropy groups in $G$ are maximal tori.)
In fact, $\fC_\reg$ is also a Poisson submanifold of $\fP^\reg_- \times \fP^\reg$ since it can
presented as the intersection of $\fP^\reg \times \fP^\reg$
with the joint zero set
of Casimir functions $F_i\in C^\infty(\fP_- \times \fP)$ of the form
\be
F_i(\cL_1,\cL_2) = C_i(\cL_1) - C_i(\cL_2), \quad \forall (\cL_1,\cL_2)\in \fP \times \fP,
\label{fCreg2}\ee
where the differentials of the functions $C_i \in C^\infty(\fP)^G$ $(i=1,\dots, r)$
span an $r$-dimensional subspace of the cotangent space of at every point of $\fP^\reg$.
For example, one may obtain the $C_i$ out of independent invariant polynomials on $\cG$ using
the exponential parametrization, $\fP = \exp(\ri \cG)$.

Next, we verify  an important consequence of Proposition \ref{prop:Pint}.

\begin{cor}\label{cor:34}
The two subrings of $C^\infty(\bM)$ defined by
\be
\fH:= \pi_2^*\left(C^\infty(\fP)^G\right)
\quad \hbox{and}\quad \fF:= \Psi^* \left(C^\infty(\fP_- \times \fP)\right)
\label{fHfF}\ee
engender a degenerate integrable system  on the symplectic
Poisson manifold $(\bM, \br{-,-}_\bM)$. The rank of this integrable system
is equal to the rank  $r=\dim(\cG_0)$ of the Lie algebra $\cG$.
\end{cor}
\begin{proof}
The elements of $\fF$ are constant along the flows of the elements of $\fH$, because $\Psi$ is constant
along those flows. Since $\Psi$ is a Poisson map, $\fF \subset C^\infty(\bM)$ is a Poisson subalgebra,
and we only have to establish the functional dimensions of $\fH$ and $\fF$.
To this end, let us denote $r:= \dim(\cG_0)$, and note that the exterior derivatives
of the elements of $C^\infty(\fP)^G$ span an $r$-dimensional space at every point $L\in \fP^\reg$.
Thus, the same is true for their $\pi_2$ pullbacks, at every point of $\pi_2^{-1}(\fP^\reg)$,
which is  a dense open submanifold of $\bM$. Hence, the functional dimension of $\fH$ is $r$.

One can verify by an easy inspection that the derivative $D\Psi$ has constant rank, equal to
$\dim(\bM) - r$, at every point of $G\times \fP^\reg$.
As a result, the transpose $(D\Psi)^*$ satisfies
\be
\dim\left( \mathrm{Im}\left( D\Psi(g,L)^*\right)\right) = \dim(\bM) - r, \qquad \forall (g,L) \in G \times \fP^\reg.
\label{presub}\ee
This implies immediately that $\fF$ has functional dimension $\dim(\bM) - r$.
If $\mu_2: \fP \times \fP$ is the projection onto the second factor, then
$\pi_2 = \mu_2 \circ \Psi$ (with $\pi_2$ in \eqref{pi2}). Thus, $\fH \subset \fF$, which completes the proof.
\end{proof}

\begin{rem}
We refer to the system of Corollary \ref{cor:34} as the \emph{integrable master system} of free
motion on the Heisenberg double.
The presentation of the Heisenberg double via its  model $\bM = G \times \exp(\ri \cG)$
highlights the close analogy with the standard degenerate integrable system on the cotangent
bundle $T^*G \simeq G \times \cG$ associated with the invariant functions $C^\infty(\cG)^G$.
Of course, we can transfer the master system to the models $M$ and $\fM$ \eqref{Heis3}
by means of the diffeomorhisms $m: M\to \bM$ and $m_2: \fM \to \bM$.
We shall make use of all three  models in what follows.
\end{rem}

\section{Hamiltonian reduction of the master system}
\label{S:4}

We consider the reduction of the master system relying on an action of $G$
on the Heisenberg double. The degenerate integrability of the reduced system will be proved
after restriction to a dense open subset of the reduced phase space.
It will be convenient to start with the model $M$ in the first subsection, but utilize the model $\bM$
for the proof of reduced integrability given in the second subsection.

\subsection{Actions of $G$ and Hamiltonian reduction}
\label{ss:41}

We begin by recalling the concept of Poisson--Lie moment map for $G$-actions \cite{LuPhD,Lu,STSrev}, adapted to our case.
Let $(\cM,\br{-,-})$ be a Poisson manifold and $\Lambda: \cM \to B$ a Poisson map with respect to the Poisson bracket $\br{-,-}_B$
\eqref{PBBT}
on the target space. Such a map can be used to generate
an (infinitesimal) Poisson action of the group $(G,\br{-,-}_G)$ \eqref{PBGT}  on $\cM$.
This works by associating a vector field $X_\cM$ on $\cM$, i.e. a derivation of $C^\infty(\cM)$,
to every $X\in \cG$ via the following formula:
 \be
df(X_\cM) := - \langle X, \br{  \Lambda, f}_\cM  \Lambda^{-1} \rangle_\bI, \qquad \forall f\in C^\infty(\cM).
\label{genmomdef}\ee
To explain the meaning of this formula, note that, for any $x_0\in \cM$, we have
\be
\{ \Lambda, f \}_\cM(x_0) := \dt \Lambda(x(t)),
\ee
where $x(t)$ is the integral curve of the Hamiltonian vector field of $f$ satisfying $x(0)= x_0$.
This yields an element of  $T_{\Lambda(x_0)} B$, which is then translated into $\cB = T_e B$ via right multiplication
by $\Lambda(x_0)^{-1}$.
According to the general theory \cite{LuPhD,Lu},  the so-obtained map $\cG\ni X \mapsto X_\cM\in \mathrm{Vect}(M)$ is an antihomomorphism
satisfying the key relation \eqref{PLid1}.
If this $\cG$-action can be integrated
to an action of $G$, then it is automatically a Poisson action.
The map $\Lambda$ is called the Poisson--Lie moment map for the pertinent action.\footnote{The equivalence of our
moment map condition \eqref{genmomdef} to the original one introduced by Lu \cite{LuPhD,Lu} follows
from Remark \ref{rem:Luconv}, since the right-invariant Maurer--Cartan 1-form  $(db) b^{-1}$ on the group $B$,
which features in \eqref{genmomdef},
becomes the left-invariant Maurer--Cartan form $b^{-1} \star d b$ on $B$ equipped with the opposite multiplication,
$b_1 \star b_2 = b_2 b_1$.
Agreement between   \eqref{genmomdef}
 with its counterpart (5.20) in \cite{STSrev} is seen by additionally noting that
 $ X_\cM$ corresponds to $-\widehat X$ used in \cite{STSrev}.}
It enjoys $G$-equivariance with respect to the dressing action \eqref{Dress} of $G$ on $B$ and the action that it generates on $\cM$.
The construction can be applied in two ways. Either one starts with a Poisson action and searches
for a corresponding (equivariant) moment map, or one starts with a Poisson map $\Lambda$ and looks for a
$G$-action that integrates
the vector fields $X_\cM$.
Since $G$ is connected and simply connected, a $G$-action
results from the infinitesimal action whenever the vector fields $X_\cM$ are complete.
As an example, one may check that the dressing action \eqref{Dress} is a Poisson action
of $(G,\br{-,-}_G)$ on $(B,\br{-,-}_B)$ with the identity map from $B$ to $B$ being the moment map.

The Heisenberg double $(M,\br{-,-}_+)$ supports three natural Poisson maps into the
Poisson--Lie group $(B,\br{-, -}_B)$. These serve as moment maps generating corresponding Poisson actions
of $(G,\br{-,-}_G$) on $M$.
In fact, $\Lambda_L$ and $\Lambda_R$ \eqref{XiLaT} are the moment maps for the $G$-actions given by left and right multiplications
of the elements of $M=G_\bC^\bR$ by the elements of the subgroup $G$, and their pointwise product
\be
\Lambda:= \Lambda_L \Lambda_R : M \to B
\label{Lambda} \ee
is the moment map for the so-called quasi-adjoint action of $G$ on $M$.
 As was shown by  Klim\v c\'\i k \cite{Kli}, the moment map \eqref{Lambda} generates a global $G$-action
 whose  action map is given explicitly by
\be
\cA^M: G \times M \to M,
\qquad
\cA^M_\eta(K):= \cA^M(\eta, K) = \eta K \Xi_R(\eta \Lambda_L(K)).
\label{actM}\ee
The map $\Lambda$ is equivariant with respect to this $G$-action on $M$ and the dressing action on $B$.
Since the center $Z(G)$ of $G$ is contained in the center of $G^\bC_\bR$, we obtain the equality
$\Xi_R(\eta \Lambda_L(K)) = \eta^{-1}$ for $\eta \in Z(G)$.
By using this, it is easily seen that
$Z(G)$ acts trivially, and the action of $G$ descends to an
effective action of adjoint group $\bar G:= G/Z(G)$ on $M$.
Since we have a Poisson action, i.e. $\cA^M$ is  a Poisson map with respect to the
product of the Poisson structures on $G$ and on $M$,
the ring of $G$-invariant functions, $C^\infty(M)^G$, is closed
under the Poisson bracket on $M$. The same is true for the functions of the moment  map; and
the Poisson algebras
\be
(C^\infty(M)^G, \br{-,-}_+)
\quad\hbox{and}\quad
(\Lambda^* C^\infty(B), \br{-,-}_+)
 \ee
 are the centralizers of each other in $(M,\br{-,-}_+)$.
 If one defines the reduced phase space by
 \be
 M^\red:= M/G,
 \label{Mred}\ee
 then the identification  $C^\infty(M^\red) \equiv C^\infty(M)^G$ equips $C^\infty(M^\red)$ with a Poisson bracket.
 In this way, one obtains the reduced Poisson space $(M^\red, \br{- ,- }_+^\red)$.
 The pullbacks by $\Lambda$ of the dressing invariant functions on $B$ engender Casimir functions on the reduced phase space,
  because
 \be
 \Lambda^*( C^\infty(B)^G) \subset C^\infty(M)^G.
 \ee
The joint level surfaces of these Casimir  functions  are unions of symplectic leaves.
The reduced phase space $M^\red$ is not a smooth manifold, but is
 a disjoint union of smooth strata.
 However, like in the smooth case, the Hamiltonian vector fields of the smooth functions on $M^\red$
 can still be obtained as projections of the Hamiltonian vector fields of the corresponding elements of
 $C^\infty(M)^G$.

 We are mostly interested in the `big cell' of $M^\red$
 that results by restriction to the principal orbit type \cite{DK,Mic,OR}  for the $G$-action.
 The principal orbits fill the dense open submanifold
 \be
 M_*:= \{ x \in M \mid G_x = Z(G)\},
 \label{M*}\ee
 where $G_x$ denotes the isotropy group of the point $x$.
  Three important features of the restriction to $M_*$ are as follows.
 First,
\be
M_*^{\red}:= M_*/G
\label{M*red}\ee
is a smooth manifold, and is a connected, dense open subset of $M^\red$.
Second, the restriction of the moment map to $M_*$ is a submersion, i.e., its derivative $D \Lambda(x): T_x M_* \to T_{\Lambda(x)} B$
is surjective at every  $x\in M_*$.
Third, $M_*$ is invariant with respect to the Hamiltonian flows of all the elements of $C^\infty(M)^G$.
These statements are immediate consequences of well known general results. For example, the third property is a consequence
of the fact that the flow $\varphi_t^H$ of any invariant function $H\in C^\infty(M)^G$ is equivariant,
\be
\cA_\eta^M \circ \varphi_t^H = \varphi_t^H \circ \cA_\eta^M, \qquad \forall \eta \in G,
\ee
which implies that the isotropy group $G_{\varphi^H_t(x)}$ is constant in $t$ for every $x\in M$.

In order to transfer the action $\cA^M$ \eqref{actM} to the alternative models $\fM$ and $\bM$ of the Heisenberg double \eqref{Heis3},
we use the relations \eqref{KdecT} and \eqref{LRrel} that lead to the identities
\be
\Lambda_L(K) =\Lambda_L(b_L g_R^{-1}) =  \Lambda_R(b_R^{-1} g_R)^{-1} =\Lambda_L(g_R^{-1} b_R)^{-1}
=: \beta_L(g_R, b_R).
\ee
The last equality is  the definition of the map $\beta_L: G \times B \to B$, which can also be written as
\be
\beta_L(g,b) = (\Dress_{g^{-1}} (b))^{-1}, \qquad \forall (g,b) \in G \times B.
\ee
Combining this with the diffeomorphisms $m_1$ \eqref{m1} and $m_2$ \eqref{m2}, we see that the Poisson action $\cA^M$
acquires the following form in terms of the models $\fM$ and $\bM$
\be
\cA^\fM_\eta(g,b) = \left( \Xi_R( \eta \beta_L(g,b))^{-1} g  \Xi_R( \eta \beta_L(g,b)), \Dress_{ \Xi_R( \eta \beta_L(g,b))^{-1}}(b)\right),
\label{cAfM}\ee
and
\be
\cA^\bM_\eta(g,L) =  \left( \tilde \eta   g \tilde \eta^{-1} , \tilde \eta L \tilde \eta^{-1} \right),
\quad \hbox{with}\quad
\tilde \eta = \Xi_R( \eta \beta_L(g,\nu^{-1}(L)))^{-1},
\label{cAbM}\ee
where we  applied the inverse of the diffeomorphism $\nu$ \eqref{nu}.
For any fixed $(g,b)$, the map
\be
\eta \mapsto \Xi_R( \eta \beta_L(g,b))^{-1}
\ee
yields  a diffeomorphism of $G$.  Consequently, the Poisson actions \eqref{cAfM} and \eqref{cAbM} are orbit equivalent (have the same orbits) as the
simpler $G$-actions given on $\fM$ and on $\bM$ by the formulae
\be
A^\fM_\eta(g,b):= (\eta g \eta^{-1}, \Dress_\eta(b)),\qquad
A^{\bM}_\eta (g,L) := (\eta g \eta^{-1}, \eta L \eta^{-1}),
\label{AfMbM}\ee
for all $\eta \in G$, $(g,b)\in \fM$ and $(g,L) \in \bM$.
These simpler actions are \emph{not} Poisson action of $G$,
 but for taking the quotients of the
respective model phase spaces they can be
 used in the same way as
their parent Poisson actions.

In Proposition \ref{prop:Pint}, we introduced the Poisson space $\fP_- \times \fP$, which is
$\fP \times \fP$ equipped with the Poisson bracket $(-1)\br{-,-}_\fP \times \br{-,-}_\fP$ with \eqref{PBonP}.
We may write the elements $(\cL_1,\cL_2) \in \fP_- \times \fP$ in the form
\be
(\cL_1, \cL_2) =  ( \nu(b_1^{-1}), \nu(b_2))
\quad \hbox{with}\quad (b_1, b_2) \in B \times B,\ee
and then we obtain a Poisson map $\hat \Lambda: \fP_- \times \fP \to B$ by the definition
\be
\hat \Lambda: (\cL_1, \cL_2) \mapsto b_1 b_2.
\label{momPP}\ee
As a  moment map,  $\hat \Lambda$  generates a Poisson action of $G$.  The action map
$\hat \cA: G \times (\fP_- \times \fP)\to \fP_- \times \fP$ operates for $\eta \in G$ by
\be
\hat\cA_\eta: (\cL_1, \cL_2) \mapsto
\left(\Xi_R(\eta b_1)^{-1} \cL_1 \Xi_R(\eta b_1), \Xi_R(\eta b_1)^{-1} \cL_2 \Xi_R(\eta b_1)\right),
\label{actonPP}\ee
using $b_1 = (\nu^{-1}(\cL_1))^{-1}$.
It is an instructive exercise to verify this statement, which we do in Appendix \ref{sec:A}.
The Poisson action \eqref{actonPP} possesses the same orbits as the alternative $G$-action having the action map
\be
\hat A_\eta: (\cL_1, \cL_2) \mapsto ( \eta \cL_1 \eta^{-1}, \eta \cL_2 \eta^{-1}).
\ee
The Poisson map $\Psi: \bM \to \fP_- \times \fP$ \eqref{Psi}
relates the relevant moment maps according to
\be
\Lambda \circ m^{-1} = \hat \Lambda \circ \Psi,
\ee
where we used $\Lambda$ in \eqref{Lambda} and $m: M\to \bM$ in \eqref{m}.
It  also satisfies the equivariance properties
\be
\Psi \circ \cA^\bM_\eta = \hat \cA_\eta \circ \Psi
\quad\hbox{and}\quad
\Psi \circ A^\bM_\eta = \hat A_\eta \circ \Psi,
\quad
\forall \eta \in G.
\ee
Our construction implies that
\be
C^\infty(\fP_- \times \fP)^G \subset C^\infty(\fP_- \times \fP)
\label{pre-fFG}\ee
is a Poisson subalgebra, and this entails that
\be
\fF^G:= \Psi^*( C^\infty(\fP_- \times \fP)^G) \subset C^\infty(\bM)^G
\label{fFG}\ee
is also a Poisson subalgebra.

 \begin{rem}
 Let us  take the opportunity to clarify a potentially confusing point that occurs in our earlier paper \cite{F2}.
 Namely,
 it was verified in Appendix C of \cite{F2} that the identity
 $A_\eta^\bM \circ \varphi_t^H = \varphi_t^H \circ A_\eta^\bM$
 \emph{does not} hold for certain $G$-invariant Hamiltonians on $\bM$.
 This is not surprising since $A_\eta^\bM$ \eqref{AfMbM} is not a Poisson action (in fact, the failure of the identity proves this),
 and the analogous identity holds if one uses the original action $\cA_\eta^\bM$ \eqref{cAbM}.
 \end{rem}

\subsection{Reduced integrability}
 \label{ss:42}

Now we use the model $\bM$ of the Heisenberg double \eqref{Heis3}, whereby
$M_*$ \eqref{M*} and $M_*^\red$ \eqref{M*red} get replaced by $\bM_*$ and $\bM_*^\red$, respectively.
That is,  with the map $m$ \eqref{m}, we have
\be
\bM_* = m(M_*), \qquad \bM_*^\red = \bM_*/G.
\label{mM*}\ee
In Section \ref{ss:32}, we  described the master integrable system $(\bM,\br{-,-}_\bM, \fH, \fF)$, whose Hamiltonians and
 constants of motion
\eqref{fHfF} were constructed relying on the Poisson map $\Psi$ \eqref{Psi}.
Eventually, we shall demonstrate that the quadruple $(\bM, \br{-,-}_\bM, \fH, \fF^G)$, with $\fF^G$ in \eqref{fFG}, engenders a degenerate
integrable system  on the Poisson manifold $\bM_*^\red$.
However, it will be advantageous to first deal with a restriction of the reduced system on a certain dense open subset
of $\bM_*^\red$, which will be found to satisfy stronger conditions than those required by Definition \ref{def:22}.

Let $\vLambda: \bM \to B$ be the moment map $\Lambda$ \eqref{Lambda} transferred to $\bM$ and $\vLambda_*$ its restriction
to $\bM_*$, i.e.,
\be
\vLambda = \Lambda \circ m^{-1}, \qquad \vLambda_{*} = \vLambda_{\vert \bM_*}.
\label{vLambda}\ee

\begin{lem}\label{lem:denseleaves} The inverse image $\vLambda_*^{-1}(B^\reg)$ is a dense open subset of $\bM_*$.
Then,
\be
\vLambda_*^{-1}(B^\reg)/G \subset \bM^\red_*
\ee
is a dense open subset, which consists of symplectic leaves of co-dimension $r= \dim(\cG_0)$.
\end{lem}
\begin{proof}
The map $\vLambda_*$ is continuous.
Since the action of $G/Z(G)$ is free on $\bM_*$, $\vLambda_*: \bM_* \to B$ is a submersion, and thus it is also an open map.
The inverse image of a dense set under an open map is dense, and the inverse image of an open set under
a continuous map is open.  Therefore,  $\vLambda_*^{-1}(B^\reg) \subset \bM_*$ is dense and open.

Let  $(\bM_*^\red, \br{-,-}_*^\red)$ denote the reduced Poisson manifold obtained by taking
the quotient of $\bM_*$ by the action of $G$. The general reduction theory \cite{Lu,OR,STSrev} says that the symplectic
leaves of this Poisson manifold are the connected components of the sets of the form
\be
\vLambda_*^{-1}(\cO_B)/G,
\label{leaf}\ee
where $\cO_B \subset B$ is a dressing orbit contained in $\vLambda_*(\bM_*)$.
The dressing orbits of maximal dimension are those that lie in $B^\reg$, and their
co-dimension is $r$.

The symplectic leaves \eqref{leaf} can also be identified as the connected components of the
level surfaces of the Casimir functions on $\bM_*^\red$ that arise from
$\vLambda^* (C^\infty(B)^G)$ restricted on $\bM_*$.  By using that both $\vLambda_*$ and the projection
$\pi: \bM_* \to \bM_*^\red$ are submersions, it is easily seen that the differentials of the Casimir functions span
an $r$-dimensional space at every point
of $\vLambda_*^{-1}(B^\reg)/G$.
Hence, the dressing orbits lying in $\vLambda_*(\bM_*) \cap B^\reg$ yield symplectic  leaves
of co-dimension $r$, which are the leaves of maximal dimension.
\end{proof}

\begin{rem}
It is known \cite{F1} that $\vLambda:\bM \to B$ is a surjective map.
Because $\vLambda$ is continuous and its restriction $\vLambda_*$ \eqref{vLambda} is an open map,
we see that
\be
B^\reg \cap \vLambda_*(\bM_*)  \subset B
\ee
is dense and open.
We suspect that $B^\reg$ is contained in $\vLambda_*(\bM_*)$,  but have not proved this.
\end{rem}

We previously introduced the `space of constants of motion' $\fC$ \eqref{fC} and its
dense subset $\fC_\reg$ \eqref{fCreg1}, which is a Poisson submanifold of $\fP_-^\reg \times \fP^\reg$.
Explicitly, $\fC_\reg$ consists of the pairs $(\cL_1,\cL_2)\in \fP^\reg \times \fP^\reg$ for which $\cL_1$ and $\cL_2$
belong to the same $G$-orbit:
\be
\fC_\reg: = \{ (g^{-1}L g, L)\mid L \in \fP^\reg,\, g\in G\}.
\label{fCreg3}\ee
 The group $G$ acts by componentwise
conjugations on $\fP^\reg \times \fP^\reg$, and this restricts to the submanifold $\fC_\reg$.
Now we introduce  $\fC_*\subset \fC_\reg$, which by definition is the dense open subset given by
the principal orbit type for the $G$-action on $\fC_\reg$.
It is easily seen that
\be
\fC_* = \{ (g^{-1} L g, L) \in \fC_\reg \mid G_{(g^{-1} L g, L)} = Z(G)\},
\label{fCstar}\ee
where  $G_{(g^{-1} L g, L)}$ is the isotropy group.
We observe that $C^\infty(\fC_*)^G$ gives rise to a Poisson structure on the smooth manifold
\be
\fC_*^\red:= \fC_*/G.
\label{fC*red}\ee
Then, we define the following subset of $\bM$:
\be
\bM_{**} := \Psi^{-1}(\fC_*).
\label{bM**}\ee
It is clear from the $G$-equivariance of $\Psi$ that $\bM_{**}$ is mapped to itself by the $G$-action.
After some preparation, our goal is to show that the restriction of the master system of
free motion on $\bM_{**}$ descends to a degenerate integrable system on the corresponding quotient space.

\begin{lem}\label{lem:44}
The inverse image $\bM_{**}$ \eqref{bM**}
is a dense open subset of $\bM_{*}\cap \pi_2^{-1}(\fP^\reg)$.
\end{lem}
\begin{proof}
The $G$-equivariance of the map $\Psi$ \eqref{Psi} implies that $G_x < G_{\Psi(x)}$ holds for all $x\in \bM$.
It follows that $G_x = Z(G)$ for all $x\in \bM_{**}$, i.e.,
\be
\bM_{**} \subset \bM_{*}\cap \pi_2^{-1}(\fP^\reg).
\ee
We know from the proof of Corollary \ref{cor:34} that the restricted map
\be
\Psi: \pi_2^{-1}(\fP^\reg) \to \fC_\reg
\ee
 is a surjective submersion.  In particular, it is both continuous and open, and therefore
the inverse image of the dense open subset $\fC_* \subset \fC_\reg$ enjoys the claimed property.
\end{proof}


\begin{rem} It is an easy consequence of what we proved that in the chain
\be
\bM_{**} \subset \bM_* \subset \bM
\label{chain1}\ee
every subset is dense and open in the one that contains it, including $\bM_{**}\subset \bM$.
Now we explain a few properties of these sets, including that
$\bM_{**} \subset \bM_*$ is a proper subset.

First, by choosing both $g$ and $L$ to be the unit element of $G^\bC_\bR$, we see that $\bM_*$ is a proper subset of $\bM$.
Next, let us recall that the center $Z(G)$ is the intersection of all maximal tori of $G$, and for a fixed
maximal torus $G_0$ one can find (see e.g. \cite{Mein}) another one, $G_0'$, such that
\be
G_0 \cap G_0' = Z(G).
\label{torsect}
\ee
Clearly, one can choose $(\tilde L, L)= \Psi(g,L)$ in such way that the isotropy subgroups with respect to the conjugation action
of $G$ on $\fP$ are $G_L= G_0$ and $G_{\tilde L} = G_0'$.  Then, the isotropy group $G_{(\tilde L, L)}$ with respect
to the diagonal conjugation action of $G$ is $Z(G)$. Consequently, $(\tilde L, L) \in \fC_*$ and  $(g,L) \in \bM_{**}$.

For concreteness, consider $G=\mathrm{SU}(n)$ and choose a pair $(g,L)$,  where $L$ is a diagonal matrix with distinct positive
eigenvalues
and $g$ is a multiple of the matrix of cyclic permutation, i.e.,
\be
g = C (E_{n,1} + \sum_{i=1}^{n-1} E_{i,i+1}),
\ee
with a constant $C$ such that $\det(g)=1$. In this case, one can verify that $G_{(g,L)} = Z(G)$, while $G_{\Psi(g,L)}$ is the
standard maximal torus of $G$. This implies that $(g,L) \in \bM_{*} \setminus \bM_{**}$.
We can generalize this example to other groups  by picking a regular $L$ whose isotropy group is the maximal torus $G_0$,
and
taking $g\in G$ to be a representative of a Coxeter element of the Weyl group
with respect to $G_0$.   By using that \cite{Mein} the fixed point set of the action of the Coxeter element on $G_0$ is the center $Z(G)$,
it is easy to verify  that for this choice  $G_{\Psi(g,L)} = Z(G)$ and $(g,L) \in \bM_* \setminus \bM_{**}$.

Consider an arbitrary  pair $(g_0', L_0)$ for which $g_0' \in G_0'$ an $L_0 \in \exp(\ri \cG_0)$ are
regular elements, i.e., their isotropy groups in $G$ are $G_0'$ and $G_0$, respectively, which are subject to \eqref{torsect}.
The $G$-orbit through $(g_0', L_0)$ belongs to $\bM_*$, and from this one sees that
\be
G^\reg \subset\pi_1(\bM_*)  \quad \hbox{and}\quad \fP^\reg \subset \pi_2(\bM_*),
\label{pibMstar}\ee
since the regular elements of $G$ (and $\fP$) are those elements whose isotropy subgroups under the conjugation action of $G$ are maximal tori,
and all maximal tori are conjugate to each other.
We also note that $\pi_2(\bM_{**}) = \fP^\reg$, but it is not clear at present if
in the relations \eqref{pibMstar} one has equalities or not.
\end{rem}

The Abelian Poisson algebra $\fH$ \eqref{fHfF}
gives rise to the reduced Abelian Poisson algebra $\fH_\red$
of Hamiltonians defined on $\bM^\red = \bM/G$. Similarly,
$\fF^G$ \eqref{fFG} descends to a Poisson algebra of functions
on $\bM^\red$, which we denote $\fF_\red$
and call the
\emph{the Poisson algebra of reduced constants of motion}.
Resulting from \eqref{chain1}, we have the chain of dense open
subsets
\be
\bM^\red_{**} \subset \bM_*^\red \subset \bM^\red,
\ee
and we let $\fH_\red^{**}$ and $\fF_\red^{**}$ denote
 the restrictions of $\fH_\red$ and $\fF_\red$ on $\bM^\red_{**} $, respectively.
 Analogously, we denote by $\fH_\red^{*}$ and $\fF_\red^*$ the corresponding restrictions on
 $\bM^\red_*$.  These Poisson algebras enjoy the natural isomorphisms
 \be
 \fH_\red^{**}  \simeq \fH_{\vert \bM_{**}},\quad
  \fF_\red^{**}  \simeq {\fF^G}_{\vert \bM_{**}}
 \quad\hbox{and}\quad
  \fH_\red^{*}  \simeq \fH_{\vert \bM_{*}},\quad
 \fF_\red^{*}  \simeq {\fF^G}_{\vert \bM_{*}}.
 \label{fHfFstars}\ee

Since $\bM_{**}$ was defined by placing  constraints on the constants of motion,
the flows of all Hamiltonians $\cH\in \fH$ \eqref{fHfF} preserve this dense open submanifold of $\bM_*$.
Taking into account that $\bM_{**}$ is also preserved by the $G$-action, we can
consider the simultaneous reduction of the pertinent Hamiltonian systems after restriction on $\bM_{**}$.
This leads to the \emph{`restricted reduced system'}  denoted
\be
(\bM_{**}^\red, \br{-,-}_{**}^\red, \fH_\red^{**}),
\label{restred}\ee
where  the Poisson structure on
$\bM_{**}^\red= \bM_{**}/G$ results from the identification
\be
(C^\infty(\bM_{**}^\red), \br{-,-}_{**}^\red) \simeq (C^\infty(\bM_{**})^G, \br{-,-}_{\bM_{**}}).
\ee

The commutative diagram of maps displayed in  Figure \ref{Diag-DIS} will be utilized for proving Theorem
\ref{thm:45} below, which represents our first main result.
Here, $p_1$  and $p_2$ are the canonical projections, $\psi:= \Psi_{\vert \bM_{**}}$ and
\be
\psi_\red: \bM_{**}^\red \to \fC_*^\red = \fC_*/G
\ee
is induced by the $G$-equivariance of the map $\psi$.
All these maps are \emph{smooth, surjective submersions}  and are \emph{Poisson maps}.
In fact, $p_1$ and $p_2$ are projections of principal fiber bundles, with structure group $G/Z(G)$.
Since the map $\Psi: \pi_2^{-1}(\fP^\reg) \to \fC_\reg$
is a surjective submersion, this property  is inherited by its  restriction $\psi$. Then, one sees by tracing the diagram   that $\psi_\red$
is also a smooth submersion.     The Poisson property of the maps
follows immediately from the definitions.
Of course, the Poisson structure on $\fC_*^\red$ is defined by the isomorphism
\be
C^\infty(\fC_*^\red) \simeq C^\infty(\fC_*)^G.
\label{PB-fCred}
\ee
Below, we shall first use the Poisson algebra
\be
\fF^\sharp_\red:= \psi_\red^*( C^\infty(\fC_*^\red)).
\label{fFred}\ee
Its relation to $\fF_\red^{**}$  \eqref{fHfFstars} will be clarified later (see Lemma \ref{lem:sharpstar}).

\begin{figure}[ht]
\centering
  \captionsetup{width=.8\linewidth}
   \begin{tikzpicture}
 \node (A)  at (-1.7,1.2) {$\bM_{**}$};
 \node (B)  at (1.7,1.2) {$\fC_{*}$};
 \node (C)  at (-1.7,-1.2) {$\bM_{ * *}^{\red}$};
 \node (D)  at (1.7,-1.2) {$\fC_{*}^{\red}$};
  \path[->] (A) edge  node[above] {$\psi$}  (B); \path[->] (B) edge node[right] {$p_2$}  (D);
  \path[->] (A) edge node[left] {$p_1$}  (C); \path[->] (C) edge node[above] {$\psi_{\red}$}  (D);
 \end{tikzpicture}
 \caption{The sets and maps used in the proof of Theorem \ref{thm:45}.
  All sets are smooth Poisson manifolds and all  maps are smooth Poisson submersions.
 $\fC_*$ is the subset of
 principal orbit type for the $G$-action on
  $\fC_\reg\subset \fC$ \eqref{fCreg3}.
 The map $\psi$ is the restriction of $\Psi$ \eqref{Psi} to $\bM_{**}= \Psi^{-1}(\fC_*)$,
  $p_1$ and $p_2$ are projections of principal fibre bundles. }
\label{Diag-DIS}
\end{figure}

The $r=1$ case, i.e. the case of $G= \mathrm{SU(2)}$, is excluded in the subsequent theorem, since in that
case the reduced system is `only' Liouville integrable.
The proof given below is similar to the proof of an analogous
statement\footnote{Incidentally, in our work we first considered the Heisenberg double;
the reduced integrability for the cotangent bundle was presented in \cite{WGMPnew}  in order to
expound the ideas in a simpler context.}
 concerning Poisson reduction of the cotangent bundle $T^*G$.

\begin{thm}\label{thm:45}
Suppose that $r=\dim(\cG_0) \neq 1$.
 Then, the restricted reduced system \eqref{restred}
 is a degenerate
  integrable system of rank $r$ with constants of motion
provided by the ring of functions
$\fF_\red^\sharp$
\eqref{fFred},
that is, the quadruple $(\bM_{**}^\red,\br{-,-}^\red_{**}, \fH_\red^{**}, \fF_\red^\sharp)$
satisfies the stipulations
of Definition \ref{def:22}.  The reduced Hamiltonian vector fields associated with $\fH_\red^{**}$ span
an $r$-dimensional subspace of the tangent space at every point of $\bM_{**}^\red$,
and the differentials of the elements of $\fF_\red^\sharp$
span a co-dimension $r$ subspace of the cotangent space.
\end{thm}
 \begin{proof}
 By the definition of the reduction, every element $\cH_\red\in \fH_\red^{**}$ obeys the relation
 \be
 \cH_\red \circ p_1 =  \cH\vert_{\bM_{**}}
 \quad\hbox{with some}\quad \cH\in \fH= \pi_2^*\left(C^\infty(\fP)^G\right),
  \ee
  and every integral curve $y(t)$ of $\cH_\red$ in $\bM_{**}^\red$ can be presented as the projection
  $p_1(x(t))$ of an integral curve $x(t)$ of $\cH$ in $\bM_{**}$.
  Since the map $\psi$ is constant along $x(t)$, we see from
  Figure \ref{Diag-DIS} that $\psi_\red$ is constant along $y(t)$. This implies immediately
  that the elements  $\fF_\red^\sharp$  \eqref{fFred} are constants of motion for every
  $\cH_\red \in \fH_\red^{**}$. The fact that $\psi_\red$ is a Poisson map entails
  that  $\fF_\red^\sharp$ forms a  Poisson subalgebra of $C^\infty(\bM_{**}^\red)$.
  Recalling that $\psi_\red$ is a submersion, we obtain
  that the dimension of the span of the differentials  of the elements of $\fF_\red^\sharp$ is equal to
  $\dim ( \fC_*^\red)$  at every point of $\bM_{**}^\red$.  Because $G/Z(G)$ acts freely on $\fC_*$ and on $\bM_{**}$, we have
  \be
  \dim ( \fC_*^\red)= \dim(G) - r  \quad\hbox{and}\quad   \dim(\bM_{**}^\red)= \dim(G),
   \ee
 which confirms the statement concerning the differentials of the elements of $\fF_\red^\sharp$.

Next, we verify the claim about the dimension of the span of the reduced Hamiltonian vector fields.
To do this, we pick an arbitrary point $y:=p_1(x) \in \bM_{**}^\red$, with some $x=(g,L)\in  \bM_{**}$.
  We know that
   the values of the reduced Hamiltonian vector fields at $y$
   result by applying the tangent map $Dp_1(x)$ to the values of the original Hamiltonian vector field at $x$.
   It follows from \eqref{Pint} that the latter are
   the tangent vectors of the form
   \be
   (V g, 0) \in T_x \bM_{**}= T_g G \oplus T_L \fP,
   \label{vect1+}\ee
   where $V\in \cG$ is given by the $\cG$-valued derivative of some function $\phi\in C^\infty(\fP)^G$ at
  $L\in \fP$.  Since $L$ is a regular element, these derivatives span a maximal Abelian subalgebra of $\cG$, and thus the linear space of
  the tangent vectors \eqref{vect1+}  is $r$-dimensional.
   We have to verify  that this linear space has zero intersection with $\mathrm{Ker}(D p_1(x))$, consisting of
    the elements
   \be
   ([W,g], [W,L])\in T_x\bM_{**}, \qquad \forall\, W\in \cG.
  \label{vect2} \ee
   Now, if two tangent vectors having the respective forms \eqref{vect1+} and \eqref{vect2} coincide, then so do their
   images in $T_{(\tilde L, L)} \fC_*$ obtained
   by the map $D\psi(x)$. But the image of the tangent vector in \eqref{vect1+}  is zero, while the image of the one in \eqref{vect2} is
   \be
   ([W, \tilde L], [W,L])\in T_{\psi(x)}\fC_* \subset  T_{\tilde L} \fP \oplus T_L\fP ,  \quad \hbox{with} \quad \tilde L = g^{-1} L g.
   \label{vect3}\ee
   Since $G/Z(G)$ acts freely $\fC_*$, the  vector in \eqref{vect3} vanishes only for $W=0$, and then the vector \eqref{vect2}
   also vanishes.
   In conclusion, the $Dp_1(x)$ image of the tangent vectors in \eqref{vect1+} has dimension $r$.

   The differentials of the elements of $\fH_\red^{**}$ span an $r$-dimensional subspace of
   $T_y^* \bM_{**}^\red$ at every $y\in \bM_{**}^\red$ since  their Hamiltonian vector fields span an $r$-dimensional
   subspace of $T_y \bM_{**}^\red$. That is, the functional dimensions of $\fH_\red^{**}$ and $\fH$ are the same.
   It is obvious that  $\fH_\red^{**}$ is contained in $\fF_\red^\sharp$.

   Lemma \ref{lem:denseleaves} implies that a dense open subset of $\bM_{**}^\red$ is filled by
   symplectic leaves of maximal dimension, which have co-dimension $r$.  If $r\neq 1$, then we have
    \be
   r < \frac{1}{2} \left(\dim(G) - r\right) = \frac{1}{2}( \dim(\bM_{**}^\red) - r).
   \ee
   This means that $r$ is strictly smaller than half the maximal dimension of the symplectic leaves
   in $\bM_{**}^\red$, and thus our restricted reduced system satisfies all conditions of Definition \ref{def:22}.

   In the excluded $r=1$ case, that is for $G= \mathrm{SU(2)}$, the reduced system is
   Liouville integrable, but there is no room for degenerate integrability in this case.
    \end{proof}

    \medskip
   Next, we prove an important consequence of  Theorem \ref{thm:45}.

  \begin{cor}\label{cor:46}
  The restriction of the system $(\bM_{**}^\red,\br{-,-}^\red_{**}, \fH_\red^{**}, \fF_\red^\sharp)$
  of Theorem \ref{thm:45} to any
  symplectic leaf of $\bM_{**}^\red$ of  co-dimension $r$ (where $r>1$)
  is a degenerate integrable system in the sense of Definition \ref{def:21}.
  \end{cor}
  \begin{proof}
  Choose $r$ functions $C_1,\dots, C_r \in C^\infty(\fP)^G$ that are functionally
  independent at every point of $\fP^\reg$.
  The restrictions of the functions $C_i \circ \vLambda_*$ on $\bM_{**}$ descend to
  Casimir functions $\mathcal{C}_i\in C^\infty(\bM_{**}^\red)$. These functions belong to
  $\fF_\red^\sharp$,
  and any symplectic leaf  of co-dimension $r$  in $\bM_{**}$ is (a connected component of)
  a joint level surface thereof.
  Now, we consider a symplectic leaf $S$ of co-dimension $r$ and  fix an arbitrary point $y\in S$.
  Then, we can select additional $(\dim(\fC_*^\red)- r)$ elements of $\fF_\red^\sharp$, say
  $f_a$, so that
  \be
  \cC_i, f_a,
  \qquad
  i=1,\dots, r,\,\,  a=1,\dots, \dim(\fC_*^\red) - r
  \ee
  are functionally independent at $y$.
  We can also find further $r$ functions, say $z_1,\dots, z_r$, so that the  functions
   \be
   \cC_i, f_a, z_i
   \label{loccord}\ee
   yield local coordinates  on an open set $U \subset \bM_{**}^\red$, containing $y$.
   It follows that the restriction of the functions $ f_a, z_i$  to the level surface $S$ of the Casimirs
   gives local coordinates on $S$ around $y\in S$.  Consequently,  the differentials of the
   restrictions  on $S$ of the elements of $\fF_\red^\sharp$  span a $\dim(S) -r$ dimensional subspace
   of the cotangent space $T^*_yS$ at every $y\in S$.
   This is all what we needed to prove, since we know from Theorem \ref{thm:45}
   that the Hamiltonian vector fields of $\fH_\red^{**}$ span an $r$-dimensional space at every
   $y\in \bM_{**}^\red$.
   They are tangent to the every symplectic leaf, and give the Hamiltonian vector fields
   of $\fH_\red^{**}$ restricted onto the leaf.
\end{proof}

Incidentally, for any generating set  $C_i$ $(i=1,\dots, r)$ of
$C^\infty(\fP)^G$,  the functions $\cH_i = C_i \circ \pi_2\in \fH$
are independent at every point of
   $\pi_2^{-1}(\fP^\reg) \subset \bM$,  and their restrictions on
   $\bM_{*}$ descend to functions  $\cH_i^\red\in C^\infty(\bM_*^\red)$, which are independent at every point of
  $(\bM_{*}\cap \pi_2^{-1}(\fP^\reg))/G$.

Finally, we wish to prove the reduced integrability on the `big cell'
$\bM_*^\red$ of the reduced phase space. We begin by recalling that $\fC_\reg$ \eqref{fCreg3}
is a regular (embedded) submanifold of $\fP^\reg \times \fP^\reg$ and $\fC_*$ is a dense open
subset of $\fC_\reg$. It follows that $\fC_*$ is also a regular submanifold of $\fP \times \fP$.
With the tautological embeddding
\be
\iota: \fC_* \to \fP \times \fP,
\label{iota1}\ee
we obtain
\be
\iota^* ( C^\infty(\fP \times \fP)^G) \subset C^\infty(\fC_*)^G.
\label{subfun}\ee
It is easy to see that this is a proper subset, since one can construct\footnote{For example, for $\fP = \exp(\ri \mathfrak{su}(n))$ take
the function $1/f$, where $f(\tilde L, L) := \prod_{i=1}^{n-1} (\lambda_i - \lambda_{i+1})$ with the $\lambda_i$ denoting the ordered
eigenvalues of $L\in \fP_\reg$.
}
smooth invariant functions on $\fC_\reg$ that blow up when approaching limit points of $\fC_\reg$ lying outside $\fP^\reg \times \fP^\reg$.
The maps in Figure \ref{Diag-DIS} give
\be
p_1^* (\fF_\red^{**})= \psi^*(\iota^* ( C^\infty(\fP \times \fP)^G))
\quad
\hbox{and}\quad
p_1^*(\fF_\red^\sharp) = \psi^*( C^\infty(\fC_*)^G),
\label{sharpstar1}\ee
where we used \eqref{fHfFstars} and \eqref{fFred}.
Then, since $p_1$ and $\psi$ are surjective submersions, we may conclude from  \eqref{subfun} that
\be
\fF_\red^{**} \subset \fF_\red^\sharp
\label{sharpstar2}\ee
is a proper subset.  Nevertheless, the following crucial lemma holds.

\begin{lem}\label{lem:sharpstar}
At each point $y\in \bM^\red_{**}$,
the differentials of the elements of $\fF_\red^{**}$ \eqref{fHfFstars} span the same linear
subspace of the cotangent space $T_y^* \bM^\red_{**}$ as do the differentials of the elements of
$\fF_\red^\sharp$ \eqref{fFred}.
\end{lem}
\begin{proof}
For any point $p \in \fC_*$, define the vector spaces
\be
\cV(p):= \mathrm{span}\{ dF(p) \mid F \in C^\infty(\fP \times \fP)^G\} < T_p^*(\fP \times \fP)
\ee
and
\be
\cW(p):= \mathrm{span}\{ dK(p) \mid K \in C^\infty(\fC_*)^G\} < T_p^*\fC_* .
\ee
Relying  on \eqref{sharpstar1}, we observe that the claim is equivalent to the equality
\be
(D\iota(p))^*(\cV(p)) =  \cW(p).
\ee
Because of \eqref{subfun}, we have
\be
  (D\iota(p))^*(\cV(p)) <  \cW(p).
  \ee
Thus, it is enough to demonstrate that these vector spaces have the same dimension.
To show this, we recall  that for any smooth action of a compact Lie group on a connected manifold $\cM$
 the dimension of the span of the
differentials of the smooth invariant functions at any $p\in \cM$ belonging the principal orbit type
is equal to the co-dimension of the orbit through $p$.
(For a proof of this  well known result, see the Appendix of \cite{WGMPnew}.)
In our case, this implies that
\be
\dim(\cV(p)) = \dim(\cW(p)) + r,
\ee
since $\dim(\cV(p))=\dim(G)$ and $\dim(\cW(p)) = \dim(\fC_*^\red)$.
Now, we notice that
the kernel of $(D\iota(p))^*$ is the span of the differentials $d F_i(p)$ $(i=1,\dots, r)$
of the functions $F_i \in C^\infty(\fP \times \fP)^G$ defined in \eqref{fCreg2}.
This follows since $p\in \fC_*$ has a coordinate neighbourhood  $U_p\subset \fP \times \fP$ whose intersection with $\fC_*$
is the joint zero set of the functions $F_i$ in $U_p$. Here, we used the description of $\fC_\reg$
outlined after \eqref{fCreg1} and that $\fC_* \subset \fC_\reg$ is an open subset.
The kernel lies in $\cV(p)$ and is of dimension $r$.
By putting these arguments together, we find that
\be
\dim((D\iota(p))^*(\cV(p)))  =\dim(\cW(p)),
\ee
which implies the claim of the lemma.
\end{proof}

Now we are ready to state the principal new result of the present paper.

\begin{thm}\label{thm:49} Suppose that $r=\rank(G)>1$ and
consider the restriction of the master system of free motion (described in Section \ref{ss:32}) on the dense, open
submanifold $\bM_* \subset \bM$ of
principal orbit type with respect to the $G$-action \eqref{AfMbM}. Then, this system  descends to
the degenerate integrable system $(\bM^\red_*, \br{-,-}_{\bM_*^\red}, \fH_\red^*, \fF_\red^*)$
on the  Poisson manifold $\bM_*^\red = \bM_*/G$, where
the Poisson subalgebras $\fH_\red^*$ and $\fF_\red^*$ of $C^\infty(\bM_*^\red) \simeq
C^\infty(\bM_*)^G$ arise from the restrictions of $\fH$ \eqref{fHfF} and $\fF^G$ \eqref{fFG}
on $\bM_* \subset \bM$, respectively.
\end{thm}
\begin{proof}
The statement follows by combining Theorem \ref{thm:45} with Lemma \ref{lem:sharpstar}.
Indeed, $\fH_\red^*$ and $\fF_\red^*$
satisfy the conditions of Definition \ref{def:22}  on $\bM_*^\red$ because of the properties  of their
restrictions $\fH_\red^{**}$ and $\fF_\red^{**}$ on the dense open
subset $\bM_{**}^\red \subset \bM_*^\red$. In particular,
Theorem \ref{thm:45} and Lemma \ref{lem:sharpstar}  imply that
the differentials of the elements of $\fF_\red^*$
span a co-dimension $r$ subspace of the cotangent space at every point of $\bM_{**}^\red$.
\end{proof}

\begin{rem}\label{rem:48}
The full reduced phase space $\bM^\red = \bM/G$ is not a smooth manifold, but it still carries
the Poisson algebra of smooth functions descending from $C^\infty(\bM)^G$.
Moreover (see \cite{OR,SL}), $\bM^\red$
can be decomposed into a disjoint union of symplectic leaves, each of which inherits
an Abelian Poisson algebra from $\fH$ \eqref{fHfF} and a Poisson algebra of constants of motion from $\fF^G$ \eqref{fFG}.
We conjecture that the reduced system is integrable on every such symplectic leaf.
It is worth noting that
the derivation \cite{FK0} of the trigonometric Ruijsenaars--Schneider model by Hamiltonian reduction
of the Heisenberg double of $\mathrm{SU}(n)$ (see also Appendic C)
 provides examples in which $\bM_*^\red$
contains symplectic leaves of dimension $2(n-1)$,  smaller than
the dimension of the generic symplectic leaves if $n\neq 2$, and the reduced system on these leaves is `only' Liouville integrable.
\end{rem}

\section{Dynamical $r$-matrix formulation of the reduced system}
\label{S:5}

Here, we first derive an explicit formula for the reduced Poisson bracket based on a convenient
partial gauge fixing.  The `gauge slice' $\bM_0$ \eqref{bM0reg} intersects every $G$-orbit contained in the dense
open submanifold $\pi_1^{-1}(G^\reg)$ of $\bM$, where $\pi_1: \bM \to G$ is the projection onto the first factor
of $\bM= G\times \fP$. The formula \eqref{RED1P} characterizes the reduced Poisson bracket since every invariant
function $\cF \in C^\infty(\bM)^G$ can be recovered from its restriction $\bar \cF$ on $\bM_0$.
Then, we describe the reduced dynamics induced on the gauge slice.

Consider the set of elements, $G_0^\reg = G_0 \cap G^\reg$,  whose centralizer in $G$ is $G_0$.
Since $G_0 < G^\bC_\bR$, the adjoint action of the elements  of $G_0$ is well-defined on $\cG^\bC_\bR$.
As is easy to see, for any $Q\in G_0^\reg$ the linear operator $(\Ad_Q-\id) \in \End(\cG^\bC_\bR)$ is invertible
on $\cG^\bC_\perp = \cG^\bC_> + \cG^\bC_<$ \eqref{triang}.
Thus,  one can define the linear operator $\cR(Q)\in \End(\cG^\bC_\bR)$  by
\be
 \cR(Q) (X) := \frac{1}{2} (\Ad_{Q} + \id)\circ (\Ad_{Q} - \id)_{\vert \cG_\perp^\bC}^{-1}(X_\perp),
 \qquad \forall Q\in G_0^\reg,\,\, \forall X= (X_0 + X_\perp) \in \cG^\bC_\bR,
 \label{RQ2}\ee
where $X_0\in \cG^\bC_0$ and $X_\perp \in \cG_\perp^\bC$
according to \eqref{triangperp}.
One can check that $\cR(Q)$ maps $\cB$ to $\cB$ and $\cG$ to $\cG$,
and (writing $\cR(Q) X:= \cR(Q)(X)$) it enjoys the identities
\be
\langle \cR(Q) X, Y \rangle_\bI = - \langle X, \cR(Q) Y\rangle_\bI, \qquad \forall X,Y\in \cG_\bR^\bC,
\ee
and
\be
\langle \cR(Q) X, Y \rangle_\bI =
\langle \cR(Q) X_\cG, Y_\cB \rangle_\bI
- \langle  X_\cB, \cR(Q) Y_\cG\rangle_\bI,
\qquad \forall X,Y\in \cG_\bR^\bC.
\label{RED1Pvar}\ee
With the aid of the exponential parametrization of $Q$ and restriction to a linear operator on $\cG$,
  $\cR(Q)$ yields the standard
trigonometric solution of the modified classical dynamical Yang--Baxter equation \cite{EV}  for the pair $\cG_0 \subset \cG$.
This dynamical $r$-matrix  features in Theorem \ref{thm:defSuthP} below.

\subsection{Reduced Poisson brackets}
 \label{ss:51}

Let us introduce
\be
\bM_0 := \{ (Q,L)\in \bM\mid Q\in G_0^\reg\}.
\label{bM0reg}\ee
The $G$-orbits that intersect $\bM_0$ fill
the  dense, open, $G$-invariant submanifold
\be
\pi_1^{-1}(G^\reg) = G^\reg \times \fP \subset \bM.
\label{bMreg}\ee
The intersection of $\bM_0$ with a $G$-orbit is an orbit
of the normalizer
\be
\fN:= N_G(G_0) = \{ \eta \in G\mid \eta G_0 \eta^{-1} = G_0 \},
\label{fN}\ee
which acts on $\bM_0$.   Here, we are referring to the action of the
group elements $\eta\in \fN < G$  determined by  equation \eqref{AfMbM}.
Colloquially, we may call $\fN$ with this action the `residual gauge group'.
Then, as is easily seen, the restriction of functions gives rise
to the following isomorphism:
\be
C^\infty(G^\reg \times \fP)^G \longleftrightarrow C^\infty(\bM_0)^\fN.
\label{isombMreg}\ee
The next definition relies on this isomorphism.

\begin{defn}\label{defn:redPB}
Let $\bar \cF, \bar \cH \in   C^\infty(\bM_0)^\fN$ be the restrictions of $\cF, \cH \in C^\infty(G^\reg \times \fP)^G$.
Then, we define
\be
\{ \bar \cF, \bar \cH \}^\red_{\bM_0} (Q,L) := \{ \cF, \cH\}_\bM(Q,L), \qquad \forall (Q,L) \in \bM_0.
\label{redbMPB}\ee
\end{defn}

On the right-hand side of \eqref{redbMPB} the restriction of
$\br{-,-}_\bM$ \eqref{A3T} to the open submanifold \eqref{bMreg} is used.
The ring of functions $C^\infty(\bM_0)^\fN$ becomes  a Poisson algebra when equipped with the
`reduced Poisson bracket' $\br{-,-}^\red_{\bM_0}$.
We shall express $\br{-,-}_{\bM_0}^\red$ in terms of the derivatives $\cD_1 \bar \cF(Q,L)\in \cB_0$ and
$\cD_2 \bar \cF(Q,L) \in \cG^\bC_\bR$, where  $\cD_1 \bar \cF(Q,L)$ is defined by
\be
\langle Y_0, \cD_1\bar\cF(Q,L)\rangle_\bI = \dt \bar\cF(e^{tY_0} Q, L),\quad\forall Y_0\in \cG_0,
\label{der157}\ee
and the derivative $\cD_2 \bar \cF$ with respect to the second variable is determined
according to \eqref{newD1} and \eqref{newD2}.

\begin{thm}\label{thm:defSuthP}
The definition \eqref{redbMPB} and equation \eqref{A3T}  imply the following formula:
\be
\{\bar \cF,\bar \cH\}_{\bM_0}^\red(Q,L) = \langle \cD_1 \bar \cF, \cD_2 \bar \cH \rangle_\bI -
\langle \cD_1 \bar \cH, \cD_2 \bar \cF \rangle_\bI
+
\langle \cR(Q)(\cD_2 \bar \cH)_\cG, (\cD_2 \bar \cF)_\cB \rangle_\bI
- \langle  (\cD_2 \bar \cH)_\cB, \cR(Q) (\cD_2 \bar \cF)_\cG\rangle_\bI.
\label{RED1P}\ee
Here, the derivatives $\cD_1 \bar \cF\in \cB_0$  and $\cD_2 \bar \cF\in \cG_\bR^\bC$ are
taken at $(Q,L)$, and $\cR(Q)$ is given by
\eqref{RQ2}.
\end{thm}
\begin{proof}
Let $\bar \cF \in C^\infty(\bM_0)^\fN$ be the restriction of $\cF \in C^\infty(G^\reg\times \fP)^G$.
To begin, note that
\be
(\cD_1  \cF(Q,L))_{\cB_0} = \cD_1 \bar \cF(Q,L)
\quad \hbox{and}\quad
\cD_2 \cF(Q,L) = \cD_2 \bar \cF(Q,L),
\quad
\forall (Q,L) \in \bM_0.
\label{DFbarF1}\ee
Then, we take the derivative at $t=0$ of
\be
\cF(e^{tY} Q e^{-tY}, e^{tY} L e^{-tY}) = \cF(Q,L), \quad \forall Y\in \cG,
\ee
and from this obtain
\be
\cD_1'\cF(Q,L) - \cD_1\cF (Q,L) = (\cD_2\cF(Q,L))_\cB,
\ee
which implies
\be
\left(\Ad_{Q^{-1}} - \id\right) \cD_1 \cF(Q,L) = (\cD_2\bar \cF(Q,L))_\cB.
\ee
This in turn entails that
\be
(\cD_2 \cF(Q,L))_{\cB_0}=0\quad\hbox{and}\quad
(\cD_1 \cF(Q,L))_{\cB_+} = - (\frac{1}{2} \id + \cR(Q)) (\cD_2 \bar \cF(Q,L))_{\cB},
\label{DFbarF2}\ee
where the subscripts refer to the decomposition of $\cB$ in \eqref{cGcBperp},
Thus, at $(Q,L) \in \bM_0$, we expressed the derivatives of $\cF$  in terms of the derivatives of $\bar \cF$.
It remains to insert these expressions, and their counterparts for $\cH$, into the
 right-hand side of  \eqref{redbMPB} given by \eqref{A3T}.

 Regarding the third term of \eqref{A3T}, we have $\left\langle Q \cD'_1\cF(Q,L) Q^{-1},  \cD_1\cH(Q,L) \right\rangle_\bI=0$,
 since $\cB$ is stable under $\Ad_Q$.
 We now use \eqref{DFbarF1} and \eqref{DFbarF2} to write
 the last two terms of \eqref{A3T}, at $(Q,L)$, as follows:
 \be
 \begin{aligned}
  \left\langle \cD_1\cF , \cD_2\cH \right\rangle_\bI
-\left\langle \cD_1 \cH , \cD_2\cF \right\rangle_\bI =&
\langle \cD_1 \bar \cF , \cD_2\bar \cH  \rangle_\bI - \langle \cD_1 \bar \cH , \cD_2\bar \cF  \rangle_\bI\\
&-\frac{1}{2} \langle (\cD_2 \bar \cF)_{\cB}, (\cD_2 \bar \cH)_\cG\rangle_\bI
+\frac{1}{2} \langle (\cD_2 \bar \cH)_{\cB}, (\cD_2 \bar \cF)_\cG\rangle_\bI  \\
 &-\langle \cR(Q)(\cD_2 \bar \cF)_{\cB}, (\cD_2 \bar \cH)_\cG\rangle_\bI
+ \langle \cR(Q)(\cD_2 \bar \cH)_{\cB}, (\cD_2 \bar \cF)_\cG\rangle_\bI   \\
=&
\langle \cD_1 \bar \cF , \cD_2\bar \cH  \rangle_\bI - \langle \cD_1 \bar \cH , \cD_2\bar \cF \rangle_\bI\\
& - \frac{1}{2} \langle \cD_2\bar \cF, (\cD_2 \bar \cH)_\cG \rangle_\bI  +
\frac{1}{2} \langle \cD_2\bar \cF, \cD_2 \bar\cH - (\cD_2 \bar \cH)_\cG\rangle_\bI \\
&+ \langle (\cD_2 \bar \cF)_\cB, \cR(Q) (\cD_2 \bar \cH)_\cG\rangle_\bI +
\langle (\cD_2 \bar \cF)_\cG, \cR(Q) (\cD_2 \bar \cH)_\cB\rangle_\bI
\\
=& \langle \cD_1 \bar \cF , \cD_2\bar \cH  \rangle_\bI - \langle \cD_1 \bar \cH , \cD_2\bar \cF  \rangle_\bI \\
 & - \langle \cD_2\bar \cF, (\cD_2 \bar \cH)_\cG \rangle_\bI + \langle \cD_2 \bar \cF, \cR(Q) (\cD_2 \bar \cH)\rangle_\bI
+\frac{1}{2} \langle \cD_2\bar \cF, \cD_2 \bar \cH \rangle_\bI.
 \end{aligned}
 \ee
 By adding also the first term of \eqref{A3T}, we get
 \be
 \{\bar \cF,\bar \cH\}_{\bM_0}^\red = \langle \cD_1 \bar \cF , \cD_2\bar \cH  \rangle_\bI
 - \langle \cD_1 \bar \cH , \cD_2\bar \cF \rangle_\bI
 + \langle \cD_2 \bar \cF, \cR(Q) (\cD_2 \bar \cH)\rangle_\bI
+\frac{1}{2} \langle \cD_2\bar \cF, \cD_2 \bar \cH \rangle_\bI.
 \ee
 Since the L.H.S. is antisymmetric  with respect to the exchange of $\bar \cF$ and $\bar \cH$, and the sum of the first
 three terms of the R.H.S. is also antisymmetric,
 we must have $\langle \cD_2\bar \cF, \cD_2 \bar \cH \rangle_\bI =0$ (as is confirmed by \eqref{PBonP}), and the
 the claim \eqref{RED1P} follows on account of the identity \eqref{RED1Pvar}.
\end{proof}

Observe that
the first term of \eqref{RED1P} contains only $(\cD_2 \bar \cH)_{\cG_0}$
since $(\cD_1 \bar \cF)_1 \in \cB _0$.  The third term depends only on $(\cD_2 \bar \cH)_{\cG_\perp}$ and on
$(\cD_2 \bar \cF)_{\cB_+}$, because $\cR(Q)$ vanishes on $\cG_0$.
 Here, we refer to the decompositions in \eqref{cGcBperp}.
Of course,  similar remarks hold for the second and fourth terms.

There are two alternative ways of dealing with the residual $\fN$ `gauge freedom' that remains after
the restriction to $\bM_0$ \eqref{bM0reg}.
The first is based on the fact that
$G_0$ is a normal subgroup of $\fN$, with the factor group being the Weyl group
\be
W_G = \fN/G_0.
\ee
This leads  to the isomorphisms
\be
\pi_1^{-1}(G^\reg)/G \simeq \bM_0/\fN \simeq (\bM_0/G_0)/W_G,
\ee
i.e., one may first take the quotient of $\bM_0$ by $G_0$ and then divide by the Weyl group.
We shall proceed in this way in the description  of a specific example in  Appendix \ref{sec:C}.
The second possibility is to take into account that $G_0^\reg$ is disconnected, and its connected components are
permuted by the Weyl group.
Therefore, one may restrict to a connected component, a so-called Weyl alcove in $G_0^\reg$,
and then there remains only the residual $G_0$ gauge symmetry. Denoting a fixed Weyl alcove by $\check G_0^\reg$,
one introduces the new gauge slice
\be
\check \bM_0 := \{ (Q,L)\in \bM\mid Q\in \check G_0^\reg\},
\label{checkbM0reg}\ee
which induces the isomorphism
\be
 C^\infty(\bM_0)^\fN \longleftrightarrow C^\infty(\check \bM_0)^{G_0}.
\label{checkisom}\ee
On the other hand, since $\bM_0$ \eqref{bM0reg} is disconnected and its  connected components are permuted by $W_G$, it is clear that
 the expression \eqref{RED1P} defines a Poisson algebra structure on the
larger ring $C^\infty(\bM_0)^{G_0}$, too.
In fact,
the Poisson bracket on $C^\infty(\bM_0)^\fN$ represents  a Poisson structure on the quotient space $\bM_0/\fN$, and this
lifts to its $W_G$ covering space $\bM_0/G_0$, giving rise to  a Poisson bracket on
$C^\infty(\bM_0)^{G_0}$.

\subsection{Reduced dynamics on the gauge slice $\bM_0$}
 \label{ss:52}

The Hamiltonian vector fields associated with the commuting Hamiltonians from $\fH$ \eqref{fHfF} are
projectable on the reduced phase space $\bM/G$, and the projected vector fields
are Hamiltonian with respect to the reduced Poisson structure.
Restricting to the dense open submanifold of $\pi_1^{-1}(G^\reg) \subset \bM$,
we  describe the vector fields on $\bM_0$ \eqref{bM0reg} whose projections on $\bM_0/\fN$ coincide
with the reduced Hamiltonian vector fields on $\pi_1^{-1}(G^\reg)/G$.
Then, we  present a construction of the corresponding integral curves on $\bM_0$.

Below, we use the terms reduced Hamiltonian vector fields and reduced dynamics on $\bM_0$.  This is a slight
abuse of terminology, since  the true reduced dynamics lives on $\bM/G$, of which $\bM_0/\fN$ is a dense open subset.

A vector field $V$ on $\bM_0 = G_0^\reg \times \fP$ can be presented as
\be
V(Q,L) = (V^1(Q,L), V^2(Q,L)) \quad \hbox{with}\quad
V^1(Q,L) \in T_Q G_0,\,\, V^2(Q,L) = T_L \fP.
\ee
Let us consider a Hamiltonian $\bar \cH \in C^\infty(\bM_0)^{\fN}$ which is the restriction of $\cH=\pi_2^* \phi \in \fH$.
This means that
\be
\bar \cH(Q,L) = \phi(L)
\quad \hbox{where} \quad
\phi\in C^\infty(\fP)^G.
\ee
For the derivatives of $\bar \cH$, we have
\be
\cD_1 \bar \cH(Q,L) = 0
\quad \hbox{and}\quad
\cD_2 \bar \cH(Q,L) = \cD \phi(L) \in \cG.
\ee
See equations \eqref{newD1} and \eqref{newD2} for the definition of $\cD \phi$.
Now, using \eqref{RED1P},
we associate with $\bar \cH$ the `reduced Hamiltonian vector field' $V_{\bar \cH}$ on $\bM_0$ by
imposing the following condition:
\be
\br{\bar \cF, \bar \cH}_{\bM_0}^\red = V_{\bar \cH}[\bar \cF],
\qquad \forall \bar \cF \in C^\infty(\bM_0)^\fN,
\label{Vdef}\ee
where $V_{\bar \cH} [\bar \cF]$ denotes the derivative of $\bar \cF$ along the vector field.
There is an ambiguity in $V_{\bar \cH}$ because of the invariance property of $\bar \cF$.
It is convenient to require \eqref{Vdef} for all $\bar \cF \in C^\infty(\bM_0)^{G_0}$, and then the
residual ambiguity in $V_{\bar \cH}$ is the addition of an arbitrary vector field that is tangent
to the $G_0$-orbits, representing an infinitesimal $G_0$ gauge transformation.

\begin{prop}
The `reduced Hamiltonian vector field' $V_{\bar \cH}$ defined above has the following form:
\be
V_{\bar \cH}^1(Q,L) =  \cD \phi(L)_0 Q, \qquad
V_{\bar \cH}^2(Q,L) = [ \cR(Q) \cD\phi(L), L],
\label{VHclaim}\ee
up to the ambiguity of adding an arbitrary vector field tangent to the $G_0$-orbits in $\bM_0$.
\end{prop}
\begin{proof}
On account of \eqref{RED1P},
the equality \eqref{Vdef} can be spelled out as
\be
\langle \cD_1 \bar \cF(Q,L), \cD \phi(L)_0\rangle_\bI +
\langle \cD_2 \bar \cF(Q,L)_\cB, \cR(Q) \cD \phi(L)\rangle_\bI =  V_{\bar \cH}[\bar \cF](Q,L).
\label{spelled1}\ee
The definition of the derivatives $\cD_i \bar \cF$ implies that \eqref{spelled1}
is equivalent to the claimed formula \eqref{VHclaim}, up to the
ambiguity of $V_{\bar \cH}$ discussed above.
\end{proof}

The  vector field $V_{\bar \cH}$ can also be derived by first taking  the original Hamiltonian vector field of $\cH\in \fH$ on $\bM$,
then restricting it to $\bM_0$, and adding a vector field tangent to the $G$-orbits in a such a way that the result is tangent to $\bM_0$.
The original Hamiltonian vector field is provided by  \eqref{dotgL}, and  one can verify  that
\be
V_{\bar \cH}^1(Q,L) = \cD\phi(L)Q + [ \cR(Q)\cD\phi(L) -\frac{1}{2} \cD \phi(L), Q],
\quad
V_{\bar \cH}^2(Q,L) =  [ \cR(Q)\cD\phi(L)  -\frac{1}{2} \cD \phi(L), L],
\ee
where the Lie brackets define an element of $T_{(Q,L)} \bM$  that is tangent to the $G$-orbit through $(Q,L)\in \bM_0$.
Note that $[\cD \phi(L), L]=0$ because of the $G$-invariance of $\phi$, and also $[\cD \phi(L)_0, Q]=0$.

Next, we present a quadrature for finding the integral curves of the vector fields \eqref{VHclaim}, which govern the dynamics
induced on the gauge slice $\bM_0$ \eqref{bM0reg}.

\begin{prop}
For a fixed function $\phi\in C^\infty(\fP)^G$ and a point $(Q_0, L_0) \in \bM_0$, let
$\eta_1(t)$ be a $G$-valued smooth function on an interval $(-\epsilon, \epsilon)\subset \bR$ such that
$\eta_1(0) = e$ and
\be
Q(t) := \eta_1(t) \exp(t \cD\phi(L_0)) Q_0\eta_1(t)^{-1} \in G_0^\reg, \quad \forall t\in (-\epsilon, \epsilon).
\label{eta1}\ee
Furthermore, let $\eta_0(t)\in G_0$ with $\eta_0(0) = e$ be the unique solution of the differential equation
\be
\dot{\eta}_0(t) \eta_0(t)^{-1} = -  \frac{1}{2} \cD \phi (\eta_1(t) L_0 \eta_1(t)^{-1})_0 - (\dot{\eta}_1(t) \eta_1(t)^{-1})_0.
\ee
Then, $Q(t)$ and $L(t):= \eta_0(t) \eta_1(t) L_0 \eta_1(t)^{-1} \eta_0(t)^{-1}$,
defined on the interval $(-\epsilon, \epsilon)$,
gives an integral curve of the vector field
$V_{\bar \cH}$ \eqref{VHclaim}, with initial value $(Q_0, L_0)$.
\end{prop}
\begin{proof}
After setting $L_1(t):= \eta_1(t) L_0 \eta_1(t)^{-1}$,  a simple calculation gives $\dot{Q}(t)  = \cD\phi(L_1(t))_0 Q(t)$ and
\be
\dot L_1(t) = [ \cR(Q(t)) \cD\phi(L_1(t)) + (\dot{\eta}_1(t) \eta_1(t)^{-1})_0 +  \frac{1}{2} \cD \phi(L_1(t))_0, L_1(t)].
\ee
To get this, one uses that $\eta_1(t) \cD \phi(L_0) \eta_1(t)^{-1} = \cD \phi(L_1(t))$ and that $\dot{Q}(t) Q(t)^{-1}\in \cG_0$ implies
\be
\left(\dot{\eta}_1(t) \eta_1(t)^{-1}\right)_\perp =\left( \cR(Q(t)) - \frac{1}{2}\right) \cD \phi(L_1(t))_\perp.
\ee
Conjugation by $\eta_0(t)$ does not change $Q(t)$, and the result follows since
$\cD \phi (L(t)) = \eta_0(t) \cD\phi( L_1(t)) \eta_0(t)^{-1}$.
\end{proof}

The first step of the above quadrature is the construction of $\eta_1(t)$ in \eqref{eta1}, which  is a purely (linear) algebraic problem.
The subsequent construction of $\eta_0(t)$ requires the calculation
of an integral,
\be
\eta_0(t) = \exp\left( -\int_0^t  \left(\dot{\eta}_1(\tau) \eta_1(\tau)^{-1}+\frac{1}{2} \cD\phi(L_1(\tau))\right)_0 d\tau \right).
\ee
This second step  can actually be omitted, since the conjugation  by $\eta_0(t)$ does not affect
the eventual projection of the integral curve on the reduced phase space.

\begin{rem}
The Hamiltonian vector field of any $\cH\in \fH$  on $\bM$, and consequently also its projection on the full reduced
phase space $\bM/G$, is complete. However, the vector field
$V_{\bar \cH}$ \eqref{VHclaim} is not complete on $\bM_0$ in general.
This is  a consequence of the fact that not all unreduced integral curves \eqref{Pint} starting in $G^\reg \times \fP$  stay in this
dense open submanifold.
\end{rem}

\section{Decoupled variables and the scaling limit}
\label{S:6}

Our first goal here is to recast the reduced Poisson bracket \eqref{RED1P} in terms of new variables consisting of canonically conjugate pairs
$q_i$, $p_i$ ($i=1,\dots, r$)
and a  `decoupled spin variable' $\lambda$. This is described by  Theorem \ref{thm:5.1}, which generalizes  a similar result
presented in \cite{FM}  for the $G=\mathrm{U}(n)$ case.  At an intermediate stage, we shall also recover the form of the reduced
Poisson bracket given in \cite{F2}.
Then,  based on Theorem \ref{thm:5.1}, we shall elaborate the connection between our reduced
systems and the well known spin Sutherland models obtained
by reduction from the cotangent bundles $T^*G$.

\subsection{Canonically conjugate pairs and `spin' variables}
 \label{ss:61}

Let us begin by noting that
in terms of the alternative model $\fM$ \eqref{Heis3} of the Heisenberg double the gauge slice $\bM_0$ \eqref{bM0reg} turns into
\be
\fM_0:= G_0^\reg \times B.
\label{fM0reg}\ee
The connection between $\fM_0$ and $\bM_0$ is given by the diffeomorphism $\bar m_2$,
which is the restriction  of the map $m_2$ \eqref{m2},
\be
\bar m_2:\fM_0 \to \bM_0, \qquad \bar m_2(Q,b)=(Q, b b^\dagger).
\label{barm2}\ee
The map $\bar m_2$ is $\fN$-equivariant with respect to the restrictions of the actions \eqref{AfMbM}, and therefore
it  induces an isomorphism
\be
C^\infty(\bM_0)^{\fN} \longleftrightarrow C^\infty(\fM_0)^{\fN}.
\label{arrow53}\ee

\begin{prop}\label{cor:43}
For functions $\bar f, \bar h \in C^\infty(\fM_0)^\fN$, the isomorphism \eqref{arrow53}
converts \eqref{RED1P} into the following equivalent formula of the reduced Poisson bracket:
\be
\{\bar f,\bar h\}_{\fM_0}^\red(Q,b) = \langle D_1 \bar f, D_2 \bar h \rangle_\bI - \langle D_1 \bar h, D_2 \bar f \rangle_\bI
+  \langle \cR(Q)(b D_2'\bar h b^{-1}), b D'_2 \bar f b^{-1}\rangle_\bI.
\label{RED1+1}\ee
Here, the derivatives are evaluated at $(Q,b)$, with $D_1 \bar f\in \cB_0$  defined
similarly to\eqref{der157}  and $D_2 \bar f, D'_2\bar f \in \cG$
defined by applying  \eqref{derB} to the second variable, and $\cR(Q)$ is given by \eqref{RQ2}.
\end{prop}
\begin{proof}
For functions $\bar \cF$ and $\bar f$ related by $\bar f = \bar \cF \circ \bar m_2$, one has
\be
\cD_1 \bar \cF = D_1 \bar f \quad\hbox{and}\quad \cD_2 \bar \cF = b D_2' \bar f b^{-1},
\ee
at the corresponding arguments, similarly to \eqref{newDrel2}. Furthermore,
$D_2 \bar f =  (b D_2' \bar f b^{-1})_\cG = (\cD_2 \bar \cF)_\cG$.
The formula \eqref{RED1+1} is obtained  from \eqref{RED1P} by direct substitution
of these relations, and their counterparts for $\bar \cH = \bar h \circ \bar m_2$.
More precisely, we also used the identity \eqref{RED1Pvar}.
\end{proof}

The formula  \eqref{RED1+1}  was obtained previously in \cite{F2} starting  from  the Poisson bracket \eqref{A3T}
on the model $\fM$ \eqref{Heis3} of the Heisenberg double.  Its equivalence with the claim of Theorem \ref{thm:defSuthP}
provides a good consistency check on our considerations.
According to the discussion at the end of Section \ref{ss:51},  the formula \eqref{RED1+1} yields a Poisson bracket also on $C^\infty(\fM_0)^{G_0}$.

Now, our goal is to recast the Poisson bracket \eqref{RED1+1} on $C^\infty(\fM_0)^{G_0}$ in terms of convenient new variables.
To do this,  we shall use that
any $b\in B$  can be decomposed uniquely as
\be
b = e^p b_+ \quad \hbox{with}\quad p \in \cB_0,\, b_+ \in B_+,
\label{bpar}\ee
and that the subgroups $B_0=\exp(\cB_0)$ and $B_+=\exp(\cB_+)$ of $B$ admit global exponential parametrization.
Our construction relies on the map
\be
\zeta: \fM_0 = G_0^\reg \times B \to G_0^\reg \times \cB_0 \times B_+
\label{zeta1}\ee
defined by the formula
\be
\zeta: (Q, e^p b_+) \mapsto (Q,p, \lambda)
\quad \hbox{with}\quad
\lambda := b_+^{-1}  Q^{-1} b_+ Q.
\label{zeta2}\ee
We remark that $\lambda$ can also be written as $\lambda =  b^{-1}  Q^{-1} b Q$.
The definition \eqref{zeta2} comes from \cite{F1}, where a `spin variable' (denoted there $S_+$) given by the same formula as our $\lambda$,
but restricted
 to the intersection of an arbitrary dressing orbit of $G$ with $B_+$, was used.
The geometric origin of the definition is expounded in Appendix \ref{sec:B}.

\begin{lem}\label{lem:55}
The map $\zeta$ \eqref{zeta1} is a  diffeomorphism. It is equivariant with respect to the $G_0$-actions for which
$\eta_0\in G_0$ sends $(Q,b)$ to $(Q, \eta_0 b\eta_0^{-1})$ and $(Q,p,\lambda)$ to $(Q,p, \eta_0 \lambda \eta_0^{-1})$.
Consequently, $\zeta$ induces an isomorphism
\be
 C^\infty(\fM_0^\reg)^{G_0}   \longleftrightarrow   C^\infty(G_0^\reg \times \cB_0 \times B_+)^{G_0}.
\label{arrow58}\ee
\end{lem}
\begin{proof}
 The properties of the map $\zeta$ were analyzed in Section 5.1 of \cite{F1}, from which the statement follows.
 (Incidentally, this paper contains an explicit formula for the inverse of $\zeta$ in the $G=\mathrm{SU}(n)$ case.)
  See also the proof of Proposition \ref{prop:65} below.
\end{proof}

Let $D_Q F $, $d_p F $ and $D_\lambda F$, $D'_\lambda F $
denote the derivatives of any real function $F \in C^\infty(G_0^\reg \times \cB_0 \times B_+)$ with  respect
to its three variables, respectively.
We have $D_Q F\in \cB_0 $, $d_p F\in \cG_0 $ and $D_\lambda F, D'_\lambda F \in \cG_\perp$.
Here, $D_Q F$ and $d_p F$ are defined by
\be
\langle Y_0, D_QF(Q,p,\lambda) \rangle_\bI + \langle X_0, D_pF(Q,p,\lambda) \rangle_\bI=
\dt F(Qe^{t Y_0},p+ t X_0, \lambda), \quad \forall X_0\in \cB_0,\, Y_0 \in \cG_0,
\ee
using that $Q e^{t Y_0} \in G_0^\reg$ for small $t$.
Recalling the decompositions \eqref{cGdec} and \eqref{cGcBperp},
$\cG_\perp$ is taken as the model of the dual space of $\cB_+$. Then, $D_\lambda F$ and $D'_\lambda F$ are defined by
\be
\langle X, D_\lambda F(Q,p,\lambda) \rangle_\bI +
\langle X', D'_\lambda F(Q,p,\lambda) \rangle_\bI =
\dt F(Q,p, e^{t X} \lambda e^{tX'}),
 \quad \forall X, X' \in \cB_+.
\label{derperp}\ee

\begin{lem}\label{lem:57}
 Consider two functions $\bar f\in C^\infty(\fM_0)$ and $F\in C^\infty(G_0^\reg\times \cB_0 \times B_+)$
 related by $\bar f = F \circ \zeta$ with the diffeomorphism \eqref{zeta2}. Then, their derivatives are connected
 according to
 \be
 D'_2 \bar f = d_p F + Q D'_\lambda F Q - (\lambda D'_\lambda F \lambda^{-1})_\cG
 \label{derid1}\ee
 and
 \be
 D_1 \bar f = D_Q F - \left(b Q D'_\lambda F Q^{-1} b^{-1}\right)_{\cB_0},
\label{derid2} \ee
where the derivatives on the left and right sides  are taken at $(Q,b)$ and at $(Q,p,\lambda)= \zeta(Q,b)$, respectively.
 The subscripts $\cG$ and $\cB_0$ refer to the decompositions \eqref{cGdec} and \eqref{cGcBperp}.
 \end{lem}
\begin{proof}
Taking any $X \in \cB_+$, and inspecting the derivatives along the curve $(Q,b \exp(tX))$ in $\fM_0$ and its $\zeta$-image in
$G_0^\reg\times \cB_0 \times B_+$, we get
\be
\langle X, D_2' \bar f(Q,b)\rangle_\bI =  \langle X, Q D'_\lambda F(Q,p,\lambda) Q^{-1} -
(\lambda D_\lambda' F(Q,p,\lambda) \lambda^{-1})_{\cG_\perp} \rangle_\bI.
\ee
Inspection of the derivatives along the curve $(Q,b\exp(t X_0))$ in $\fM_0$, with $X_0\in \cB_0$, gives
\be
\langle X_0, D_2' \bar f(Q,b)\rangle_\bI =  \langle X_0,  d_p F(Q,p,\lambda) -
(\lambda D_\lambda' F(Q,p,\lambda) \lambda^{-1})_{\cG_0} \rangle_\bI.
\ee
Together, these imply the equality \eqref{derid1}.
To derive \eqref{derid2}, we use the identity
\be
\bar f(Q e^{tY_0}, b) = F(Q e^{t Y_0}, p, \lambda Q^{-1} b^{-1} e^{-t Y_0} b Q e^{tY_0}), \quad \forall Y_0\in \cG_0,
\ee
and note that
\be
\dt \left(Q^{-1} b^{-1} e^{-t Y_0} b Q e^{tY_0}\right) = \left(Y_0 - Q^{-1} b^{-1} Y_0 b Q\right)\in \cB_+.
\ee
Consequently, at the arguments related by $\zeta$,  we get
\be
 \langle Y_0, D_1 \bar f \rangle_\bI =  \langle Y_0, D_Q F\rangle_\bI +
\langle Y_0 - Q^{-1} b^{-1} Y_0 b Q, D'_\lambda F \rangle_\bI= \langle Y_0,  D_Q F
- \left( b Q D_\lambda' F Q^{-1} b^{-1}\right)_{\cB_0}\rangle_\bI,
\ee
which completes the proof.
\end{proof}

Since any two functions $F,H \in C^\infty(G_0^\reg \times \cB_0 \times B_+)^{G_0}$ are related to unique functions
$\bar f, \bar h \in C^\infty(\fM_0)^{G_0}$
by
\be
F \circ \zeta = \bar f, \quad H\circ \zeta = \bar h,
\ee
 we can define $\{F,H\}_0^\red \in C^\infty(G_0^\reg\times \cB_0 \times B_+)^{G_0}$ by the requirement
\be
\{ F,H\}_0^\red \circ \zeta := \{ \bar f, \bar h\}^\red_{\fM_0}.
\label{FHred-def}\ee

\begin{thm}\label{thm:5.1}
In terms of the new variables introduced via the map $\zeta$ \eqref{zeta2}, using the definition \eqref{FHred-def},
the reduced Poisson bracket \eqref{RED1+1}
acquires the `decoupled form'
\be
\{F,H\}^\red_0(Q,p,\lambda) =
\langle D_Q F, d_p H \rangle_\bI - \langle D_Q H, d_p F \rangle_\bI
+  \langle \lambda D_\lambda' F \lambda^{-1},  D_\lambda  H \rangle_\bI,
\label{FHred-form}\ee
where the derivatives of $F, H\in C^\infty(G_0^\reg\times \cB_0 \times B_+)^{G_0}$
are taken at $(Q,p,\lambda)$.
The functions of the form $F(Q,p,\lambda) = \varphi(\lambda)$ with $\varphi \in C^\infty(B)^G$ are in the center of this Poisson bracket.
 \end{thm}
\begin{proof}
The formula \eqref{FHred-form} follows from the direct substitution of the relations of  Lemma \ref{lem:57} into \eqref{RED1+1}.
The required tedious manipulations leading  from \eqref{RED1+1} to \eqref{FHred-form} are omitted, since they essentially coincide with the
calculation presented in \cite{FM}, where the $G=\mathrm{U}(n)$ analogue of the formula was derived.

If $F$ depends only on $\lambda$ and is the restriction of $\varphi \in C^\infty(B)^G$, then we have
\be
D'_\lambda F = D'\varphi (\lambda)  - X_0
\quad \hbox{with}\quad X_0 = (D' \varphi(\lambda))_0.
\ee
In this case,  since $\lambda D' \varphi(\lambda) \lambda ^{-1} \in \cG$ by \eqref{invprop}, we can write
\be
\{F,H\}^\red_0(Q,p,\lambda)
= \langle X_0 - \lambda X_0 \lambda^{-1},  D_\lambda  H \rangle_\bI.
\ee
Notice that $Y:= X_0 - \lambda X_0 \lambda^{-1}$ belongs to $\cB_+$  and
\be
\dt e^{t Y} \lambda = \dt e^{t X_0} \lambda e^{-t X_0}.
\ee
This identity and the fact that $H$ is $G_0$-invariant imply
\be
\langle Y,  D_\lambda  H \rangle_\bI = \dt H(Q,p, e^{t Y} \lambda) = \dt H(Q,p, e^{tX_0} \lambda e^{-t X_0}) =0,
\ee
which completes the proof.  Incidentally, an alternative proof can be obtained starting from the reduced symplectic form
derived in \cite{F1}.
 \end{proof}

\begin{rem}
As it was discussed around \eqref{checkbM0reg},  the variable
$Q$ may be restricted to a Weyl alcove $\check G_0^\reg$. By using \eqref{checkisom} and the relevant restriction of the map $\zeta$ \eqref{zeta1}, we
obtain the isomorphism
\be
 C^\infty( \fM_0)^{\fN}  \longleftrightarrow   C^\infty(\check \fM_0)^{G_0}   \longleftrightarrow
 C^\infty(\check G_0^\reg \times \cB_0 \times B_+)^{G_0}
\label{arrow58+}\ee
with  $\check \fM_0 = \check G_0^\reg \times B$.
This may be a preferred way to proceed, since we do not have an explicit formula for the action of the group $\fN$
on  $G_0^\reg \times \cB_0 \times B_+$.  Such an action is determined by transferring the action \eqref{AfMbM} of $\fN < G$ on $\fM_0$
via the diffeomorphism  $\zeta$, but its explicit form is not easy to find.
According to Theorem \ref{thm:5.1},
the additional restriction of the variable $\lambda$ to the intersection of $B_+$ with a dressing orbit can be achieved
by fixing Casimir functions.   In fact, this leads to a Poisson subspace of $(\check G_0^\reg \times \cB_0 \times B_+)/G_0$
and the  Poisson bracket on this subspace  corresponds to the reduced symplectic form exhibited in \cite{F1}.
\end{rem}

\begin{rem}
With very small modifications, all results of the paper are valid  for non-Abelian reductive  Lie groups, too, and
the simply connectedness of $G$ is also not essential.
In the paper \cite{F3} we dealt with the important example for which
\be
G= \mathrm{U}(n), \quad G^\bC = \mathrm{GL}(n,\bC), \quad B= \mathrm{B}(n),
\ee
where $\mathrm{B}(n)$ consists of those upper triangular elements of $\mathrm{GL}(n,\bC)$ whose diagonal entries are real, positive numbers.
The bilinear form on the real Lie algebra $\gl$ can be taken to be
\be
\langle X, Y \rangle_\bI = \Im \tr(XY), \quad\forall X,Y \in \gl.
\label{glform}\ee
In this case $\cG = \u(n)$, and $\fP = \exp(\ri \u(n))$ is the set of positive definite, Hermitian matrices, which
is an \emph{open} subset of the vector space $\ri \u(n)$ of Hermitian matrices.
The vector spaces $\ri \u(n)$ and $\u(n)$ are in duality
with respect to bilinear form \eqref{glform}.
Thus, for any function $\cF \in C^\infty(\fP)$, one can define its $\u(n)$-valued differential $d\cF$ by the
requirement
\be
\langle Z, d\cF(L)\rangle_\bI = \dt \cF(L + t Z), \qquad \forall Z\in \ri \u(n),
\ee
since $(L + t Z)$ belongs to $\fP$ for sufficiently small $t$.
Relating $\cF \in C^\infty(\fP)$ to $\varphi \in C^\infty(B)$ by
\be
\cF (bb^\dagger) = \varphi(b), \qquad \forall b \in \mathrm{B}(n),
\ee
one can verify the following identity:
\be
 2 L d\cF(L) = b D'\varphi(b) b^{-1}  \qquad\hbox{for}\quad L = b b^\dagger.
\ee
By applying the counterpart of this identity  to related functions defined on $\bM$ and $\fM$ \eqref{Heis3},
and on $\bM_0$ \eqref{bM0reg} and $\fM_0$ \eqref{fM0reg},
the formulas  \eqref{A2T}
and \eqref{RED1+1} are converted into those derived in \cite{F3}, which served as the starting  in the joint paper
with Marshall \cite{FM}.
\end{rem}

\subsection{Connection  with  spin Sutherland models}
 \label{ss:62}

It is easily seen from the formula \eqref{PBGT} that $G_0 < G$ is a Poisson--Lie subgroup
on which the Poisson bracket  vanishes.  This entails that the restriction
 of the dressing action
of $G$ on $B$ yields a classical Hamiltonian action of $G_0$.
This action operates by conjugations,
\be
\Dress_{\eta_0}(S) = \eta_0 S \eta_0^{-1},
\qquad \forall \eta_0\in G_0,\, S\in B,
\ee
and is generated
by the classical moment map $S \mapsto \log S_0 \in \cB_0\simeq \cG_0^*$, which is
defined by applying the decomposition $S = S_0 \lambda$ with $S_0\in B_0$ and $\lambda \in B_+$.
A particular reduction of the Poisson space $(B,\br{-,-}_B)$ is obtained by setting this moment map to zero.
Identifying the smooth functions on $B_+/G_0$ with $C^\infty(B_+)^{G_0}$,
the third term of the Poisson bracket \eqref{FHred-form}  represents precisely the reduced Poisson
structure  arising from $(B,\br{-,-}_B)$ in this way.

The Poisson--Lie group $(B,\br{-,-}_B)$ is a nonlinear analogue of
$(\cG^*,\br{-,-}_{\cG^*})$ equipped with the linear Lie--Poisson bracket.
Using the pairing \eqref{impair}, $\cB$ can be taken as the model of the dual space $\cG^*$, and its reduction
with respect to $G_0$ at the zero value of the moment map gives a Poisson structure
on $C^\infty(\cB_+)^{G_0}$.
This is a building block  of a linear counterpart
of the Poisson bracket \eqref{FHred-form}.
Denoting the elements of $G_0^\reg \times \cB_0 \times \cB_+$ by $(Q,p,X)$,
the Poisson bracket  at issue has the form
\be
\{f,h\}_{\lin}(Q,p,X) =
\langle D_Q f, d_p h \rangle_\bI - \langle D_Q h, d_p f \rangle_\bI
+  \langle X, [ d_X f, d_X h]   \rangle_\bI,
\label{FHred-lin}\ee
where the derivatives of $f,h \in C^\infty(G_0^\reg \times \cB_0 \times \cB_+)^{G_0}$
are taken at $(Q,p,X)$, and
$d_X f\in \cG_\perp\simeq (\cB_+)^*$ denotes the
differential of $f$ with respect to its third variable.

We recall (see e.g. \cite{Re1}) that the Poisson algebra
\be
( C^\infty(G_0^\reg\times\cB_0\times \cB_+)^{G_0}, \br{-,-}_\lin)
\label{linred}\ee
encodes  the reduced Poisson structure obtained by taking the quotient
of the cotangent bundle $T^*G$ with respect
to the  conjugation action of $G$.
More precisely, this is true for the dense open subset $T^* G^\reg/G\subset T^*G/G$,
after  further restricting the variable $Q$ to a Weyl alcove $\check G_0^\reg \subset G_0^\reg$ (or
considering only those functions that are invariant with respect to the normalizer $\fN$ \eqref{fN}).
The cotangent bundle $T^*G$ carries the degenerate integrable
system whose main Hamiltonian is the kinetic energy corresponding
to the bi-ivariant Riemannian metric on $G$.
The reduction of the kinetic energy yields the
spin Sutherland  Hamiltonian \cite{EFK,Re1}, represented by the following element of the Poisson algebra \eqref{linred}:
\be
H_{\mathrm{spin-Suth}}(e^{\ri q}, p, X) = \frac{1}{2} \langle p, p \rangle + \frac{1}{8}  \sum_{\alpha \in \Phi^+}
 \frac{1}{\vert \alpha \vert^2}   \frac{  \vert X_\alpha  \vert^2}{ \sin^2(\alpha(q)/2)}
\quad \hbox{with}\quad X= \sum_{\alpha \in \Phi^+} X_\alpha E_\alpha\in \cB_+.
\label{Suth}\ee
Here, $\Phi^+$ denotes the set of positive roots of $\cG^\bC$,  so that $\cB_+$ is the complex span
of the root vectors $E_\alpha$ for $\alpha \in \Phi^+$.
We employ
 a Weyl--Chevalley  basis \cite{Sam} of $\cG^\bC$, for which  $E_{-\alpha} =- \theta( E_\alpha)$ with the Cartan involution
 $\theta$ \eqref{theta}, and
  $\langle E_\alpha, E_{-\alpha}\rangle= 2/\vert \alpha \vert^2$ holds.
  It is worth stressing that in \eqref{Suth}  $\cB$ is taken as the model of $\cG^*$.

The next result clarifies the connection between the Poisson algebras of the spin Sutherland models and our models.

\begin{prop}\label{prop:64}
For any real $\epsilon \neq  0$,
let us define the $G_0$-equivariant diffeomorphism
\be
\mu_\epsilon: G_0^\reg \times \cB_0 \times \cB_+ \to G_0^\reg \times \cB_0 \times B_+,
\quad
\mu_\epsilon: (Q,p, X) \mapsto (Q, \epsilon p, \exp(\epsilon X)).
\label{mueps}\ee
Then, the `linear Poisson structure' \eqref{FHred-lin} is related to the nonlinear one  \eqref{FHred-form}
according to the formula
\be
\br{f,h}_\lin = \lim_{\epsilon \to 0} \epsilon \br{ f \circ \mu_\epsilon^{-1}, h \circ \mu_\epsilon^{-1}}^\red_0 \circ \mu_\epsilon.
\label{scalePB}\ee
\end{prop}
\begin{proof}
To keep the formulae short, let us focus on functions that do not depend on $Q$ and $p$.
Choosing a basis $\{ T^a\}$ of $\cG_\perp$, we may use the components
\be
X^a = \langle X,T^a\rangle_\bI \quad\hbox{and}\quad
\sigma^a:= \langle \sigma, T^a\rangle_\bI \quad \hbox{with}\quad  \lambda = e^\sigma\in B_+
\ee
as coordinates on the respective spaces $\cB_+$ and $B_+$.
 Both formulae \eqref{FHred-form} and \eqref{FHred-lin}
define bi-derivations (bivector fields), which can be used
to  calculate the brackets of arbitrary smooth functions (not only $G_0$-invariant ones).
Adopting the usual summation convention, we can write
\be
\{f,h\}_\lin = \Pi_\lin^{a,c} \partial_a f \partial_c h,
\qquad
\Pi_\lin^{a,c}(X) = \langle X, [T^a, T^c] \rangle_\bI,
\ee
where $\partial_a$ denotes partial derivative with respect to the coordinate $X^a$.
On the other hand, we obtain
\be
\{F,H\}_0^\red = \Pi_0^{a,c} \partial_a F \partial_c H,
\qquad
\Pi_0^{a,c}(\sigma) =  \Pi_\lin^{a,c}(\sigma) + \cP^{a,c}(\sigma),
\label{Pi0}\ee
where $\cP^{a,c}(\sigma)$ is a polynomial in the components of $\sigma$, whose lowest order terms are quadratic.
Here,  the partial derivatives are with respect to the coordinates $\sigma^a$.
The formula \eqref{Pi0} follows from \eqref{FHred-form} by means of a simple calculation of
\be
\Pi_0^{a,c}(\sigma) := \langle \lambda D'_\lambda \sigma^a \lambda^{-1}, D_\lambda \sigma^c \rangle_\bI.
\ee
The point is to notice from \eqref{derperp} that the derivatives of the coordinate functions $\sigma^a$ have the form
\be
D_\lambda \sigma^a = T^a + \cP^a(\sigma), \quad
D'_\lambda \sigma^a = T^a + {\cP'}^a(\sigma),
\ee
with certain $\cG_\perp$-valued functions  $\cP^a(\sigma)$ and ${\cP'}^a(\sigma)$ whose components
are multivariable polynomials in the coordinate functions, \emph{without constant terms}.
By using the chain rule, and noting that the partial derivatives of the components of $\mu_\epsilon^{-1}$ give
$\epsilon^{-1}$-times the unit matrix,
\eqref{Pi0} leads to
\be
\left(\epsilon \br{ f \circ \mu_\epsilon^{-1}, h \circ \mu_\epsilon^{-1}}^\red_0 \circ \mu_\epsilon\right)(X)
= \frac{1}{\epsilon} \left( \Pi_\lin^{a,c}(\epsilon X) + \cP^{a,c}(\epsilon X) \right) (\partial_a f)(X) (\partial_c h)(X).
\ee
This implies the claim \eqref{scalePB} for functions $f, h$ that depend only on $X\in \cB_+$.
The possible dependence  on $Q$ and $p$ is taken into account effortlessly.
\end{proof}

In view of Proposition \ref{prop:64}, we say that the `linear structure' \eqref{FHred-lin} is the \emph{scaling limit} of
the nonlinear one \eqref{FHred-form}. Notice that in \eqref{scalePB} the bracket $\br{-,-}_0^\red$ has also been rescaled by $\epsilon$.
We put `linear Poisson structure' in quotation marks, since we are dealing with Poisson brackets
of $G_0$-invariant functions, and  no linear function of $X\in \cB_+$ is $G_0$-invariant.

Finally, we explain  how the spin Sutherland Hamiltonian \eqref{Suth} can be recovered  from specific Hamiltonians of our reduced  system
obtained from the Heisenberg double.
For this purpose, we take an arbitrary finite dimensional irreducible representation
$\rho: G^\bC \to \mathrm{SL}(V)$, and introduce an inner product on the complex vector space $V$ in such a way that we have,
\be
\rho(K^\dagger) = \rho(K)^\dagger, \qquad \forall K \in G^\bC,
\ee
where $K^\dagger$ is defined in \eqref{taumap} and $\rho(K)^\dagger$ denotes adjoint with respect to the inner product.
This ensures that $G$ and $\fP$ are represented by unitary and by positive Hermitian operators, respectively.
Then, the character of the representation gives the element $\phi^\rho\in C^\infty(\fP)^G$,
\be
\phi^\rho(L) := \tr_\rho(L) := c_\rho \tr \rho(L),
 \qquad \forall L\in \fP.
\ee
Here,  $c_\rho$ is a (positive)  normalization constant chosen in such a way that
the trace taken in the representation reproduces the Killing form,  i.e.,
\be
  \langle X,  Y \rangle = c_\rho \tr ( \rho(X) \rho(Y)),
\qquad \forall X,Y\in \cG^\bC,
\ee
where the Lie algebra representation is also denoted by $\rho$.

Now, let us express $L$ in terms of the decoupled variables $(Q,p, \lambda)$ introduced in equation \eqref{zeta2},
with $\lambda\in B_+$ written as $\lambda = \exp(\sigma)$. This yields the  Hamiltonian
\be
H^\rho(Q, p, \sigma):= \tr_\rho (L(Q,p,\sigma))
\quad\hbox{with}\quad
L(Q,p,\sigma) = e^p b_+(Q,\sigma) b_+(Q, \sigma)^\dagger e^p,
\label{Hrho}\ee
where  $b_+(Q,\sigma)$ is determined by  the relation
\be
b_+^{-1} Q^{-1} b_+ Q = e^\sigma.
\label{sigmarel}\ee

\begin{prop}\label{prop:65}
The spin Sutherland Hamiltonian \eqref{Suth} is the scaling limit of $H^\rho$ \eqref{Hrho} as follows:
\be
H_{\mathrm{spin-Suth}}= \lim_{\epsilon \to 0} \frac{1}{4\epsilon^2} ( H^\rho \circ \mu_\epsilon - c_\rho \dim_\rho).
\label{Hamscale}\ee
Here, we use the map $\mu_\epsilon: (Q,p, X) \mapsto (Q,\epsilon p, \epsilon X)$, which is just  \eqref{mueps} written
in terms of the exponential parametrization of $B_+$.
\end{prop}
\begin{proof}
The proof  is based on calculations that appeared in \cite{F1} (without the interpretation as a scaling limit).
Let us employ the parametrizations
\be
b_+ = \exp(\beta),
\quad
\lambda = \exp(\sigma)
\quad\hbox{with} \quad \beta = \sum_{\alpha\in \Phi^+} \beta_\alpha E_\alpha,\quad \sigma = \sum_{\alpha\in \Phi^+} \sigma_\alpha E_\alpha
 \ee
 and spell out the relation \eqref{sigmarel} as
 \be
\exp(-\beta + Q^{-1} \beta Q + \frac{1}{2} [ Q^{-1} \beta Q, \beta] + \cdots ) = \exp(\sigma),
\label{F3}\ee
which results from the Baker--Campbell--Hausdorff formula.
The dots denote higher `commutators', of which there appear only finitely many, for $\cB_+$ is nilpotent.
Since the exponential map from $\cB_+$ to $B_+$ is a diffeomorphism, one may use \eqref{F3}
to establish the form of the dependence of $\beta$ on  $\sigma$ and $Q= e^{\ri q}$.
With the aid of induction according to the height of the roots, one finds \cite{F1} that
\be
\beta_\alpha = \frac{\sigma_\alpha}{e^{-\ri \alpha(q)} -1 } + \Gamma_\alpha(e^{\ri q}, \sigma),
\label{F5}\ee
where $\Gamma_\alpha$ has the form
\be
\Gamma_\alpha = \sum_{k\geq 2} \sum_{\varphi_1,\dots, \varphi_k} f_{\varphi_1,\dots, \varphi_k}(e^{\ri q})
\sigma_{\varphi_1}\dots \sigma_{\varphi_k}.
\label{F6}\ee
The sum is taken over
the unordered collections  $\varphi_1,\dots, \varphi_k$ of
positive roots,  with possible repetitions, such that $\alpha = \varphi_1 + \dots + \varphi_k$.
The  coefficients $f_{\varphi_1,\dots, \varphi_k}$ are rational functions in
$e^{\ri \alpha_1(q)},\dots,  e^{\ri \alpha_r(q)}$, where $\alpha_1,\dots, \alpha_r$ are the simple roots.
By substituting \eqref{F5} of $\beta_\alpha$ into \eqref{Hrho},  and expanding $\rho(b_+) = \exp(\rho(\beta))$,
one obtains the formula
 \be
H^\rho(e^{\ri q},p,\sigma)= c_\rho \tr \left(e^{2\rho(p)}\left(\1_V + \frac{1}{4}\sum_{\alpha\in \Phi^+}
\frac{ \vert \sigma_\alpha \vert^2 \rho(E_{\alpha})\rho( E_{-\alpha})}{\sin^2(\alpha(q)/2)} +
\o_2(\sigma,\sigma^*) \right)\right).
\label{F8}\ee
Here, $\o_2(\sigma,\sigma^*)$ stands for a finite number of terms that have total degree  3 and higher in the components of $\sigma$ and their complex
conjugates. These terms depend also on $Q = \exp(\ri q)$, and $\1_V$ denotes the identity operator on the representation space $V$.
To get \eqref{F8}, we used that $\tr \left(\rho(E_\alpha)\rho( E_{-\gamma}) e^{2\rho(p)}  \right) =0$ unless $\gamma= \alpha$.
 Then, by expanding $e^{2\rho(p)}$ as well and noting that $\tr(\rho(p)) = 0$ because $\cG^\bC$ is simple,
 the claim \eqref{Hamscale} follows from \eqref{F8}.
\end{proof}

\begin{rem}
The standard spin Sutherland Lax matrix can be  recovered
as  the scaling limit of  our Lax matrix $L(Q,p,\sigma)$  \eqref{Hrho}.
It can be shown that the Hamiltonians  $H^\rho$ \eqref{Hrho} corresponding to the $r$ fundamental highest weight
representations of $G^\bC$ are functionally independent on a dense open subset.
Motivated by the presence of the factor $e^{2p}$ in \eqref{F8}  and the relation \eqref{Hamscale},
 $H^\rho$ \eqref{Hrho} may be called a Hamiltonian
of spin Ruijsenaars--Schneider type.
We explain in Appendix \ref{sec:C}  that
a special case of these
Hamiltonians for $G=\mathrm{SU}(n)$ reproduces  the standard (spinless) trigonometric Ruijsenaars--Schneider Hamiltonian \cite{RS}
 on a symplectic leaf of the reduced phase space.
\end{rem}

\section{Discussion}
\label{S:7}

In this work we continued our previous investigations  \cite{F1,F2} of  Poisson--Lie analogues
of  spin Sutherland models.
We solved an important
open question regarding the integrability of these models, and further developed various aspects of the earlier results.

 Reduced integrability was argued in \cite{F1,F2} by exhibiting a large set of constants of motion,
but the precise  counting and other technical details were missing.
Our principal new results are given by Theorem \ref{thm:45} with Corollary \ref{cor:46} and Theorem \ref{thm:49}.
 Theorem \ref{thm:49} states the degenerate integrability of our models on the Poisson manifold
 $\bM_*^\red$, which is   the smooth component of the reduced phase space corresponding to the principal orbit type for
 the underlying group action on the  Heisenberg double. Theorem \ref{thm:45}
 establishes even stronger properties of  the restricted reduced system on the dense open subset $\bM_{**}^\red\subset \bM_*^\red$
 associated with the subset $\fC_*$ \eqref{fCstar} of principal orbit type in the space of unreduced constants of motion.  Corollary  \ref{cor:46}
 deals with the generic symplectic leaves in $\bM_{**}^\red$.

In addition to the thorough analysis of reduced integrability, we also presented a novel description of the reduced Poisson brackets.
This is given by Theorem \ref{thm:defSuthP},  which was derived
utilizing the model $\bM= G\times \fP$ \eqref{fP} of the Heisenberg double developed in this paper.
Then, in Theorem \ref{thm:5.1}, we expressed the reduced Poisson brackets in terms of canonically conjugate pairs
and decoupled spin variables, and subsequently used this
to deepen the previously found \cite{F1} connection between our models and the standard spin Sutherland models.
The latter models are recovered via the scaling limit characterized by Propositions \ref{prop:64} and \ref{prop:65}.

Turning to open problems, we wish to stress that
further work is required to clarify the integrability properties of the restrictions of the reduced systems  on arbitrary
 symplectic leaves of the full reduced phase space. This is true concerning  both  the spin Sutherland models and their Poisson--Lie deformations.
Other challenging  problems concern the quantization and the construction of elliptic counterparts of our trigonometric systems.
It is well known that the spin Sutherland models can be quantized by combining harmonic analysis on the underlying Lie groups with quantum Hamiltonian
reduction \cite{EFK,FP2}, and it should be possible to generalize this to our systems.

Throughout the paper, we strove for a careful exposition of  the nontrivial technical issues in the hope
that the resulting text may serve as a useful  starting point for future studies
 of   open problems of the subject.
 The auxiliary  material of the appendices is included having the same goal in mind.

\subsection*{Data availability statement}
No new data were created or analysed in this study.

\bigskip
\subsection*{Acknowledgement}
I wish to thank  Maxime Fairon for  useful remarks on the manuscript.
This work was supported in part by the NKFIH research grant K134946.
The publication was also supported by
the University of Szeged Open Access Fund, under Grant Number 6814.

 \appendix

\section{The Poisson action on $\fP_- \times \fP$}
\label{sec:A}

Here, we sketch the derivation of the Poisson action \eqref{actonPP} of $G$ on $\fP_- \times \fP$.
We proceed by first deriving an equivalent action on $B\times B$, which we then
transfer to $\fP_- \times \fP$ by means of the Poisson diffeomorphism  $\mu: B\times B \to \fP_- \times \fP$
given (with $\nu$ in \eqref{nu}) by
\be
\mu: (b_1, b_2) \mapsto (\nu(b_1^{-1}), \nu(b_2)).
\label{mu}\ee
Our reasoning illustrates how one may find the action
starting from a Poisson--Lie moment map.

To begin, we note from \eqref{PBBT} that the Hamiltonian vector field, $V_F$, associated
with $F\in C^\infty(B\times B)$ by means of the product Poisson structure
reads
\be
V_F(b_1,b_2) = \left( b_1 (b_1^{-1} D_1 F(b_1,b_2) b_1)_\cB, b_2 (b_2^{-1} D_2 F(b_1,b_2) b_2)_\cB \right),
\ee
where $D_1F$ and $D_2F$ are the derivatives with respect to the first and second variable, respectively.
Next, we define the Poisson map $J: B\times B \to B$ by $J(b_1,b_2):= b_1 b_2$, and from $\br{F,J} = - V_F[J]$ find
\be
\langle X, \br{ F, J}_{B\times B}(p) J(p)^{-1} \rangle_\bI= \langle (b_1^{-1} X b_1)_\cB, D_1'F(p)\rangle_\bI
+\langle (b_2^{-1}(b_1^{-1} X b_1)_\cG b_2)_\cB, D_2'F(p)\rangle_\bI,
\ee
at any  $p=(b_1, b_2) \in B\times B$, for any $X\in \cG$.
This means that the vector field $X_{B\times B}$ generated by the moment map $J$ has the form
\be
X_{B\times B}(b_1,b_2) = \left( \dress_X(b_1), \dress_{(b_1^{-1} X b_1)_\cG} (b_2)\right).
\ee
We claim that this is the infinitesimal form of the $G$-action on $B\times B$ defined by the maps
\be
\cA_\eta(b_1, b_2) = \left(\Dress_\eta(b_1), \Dress_{\Xi_R(\eta b_1)^{-1}}(b_2)\right), \qquad \forall \eta\in G.
\label{actapp}\ee
The action property $\cA_{\eta_1} \circ \cA_{\eta_2} = \cA_{\eta_1 \eta_2}$ is proved by using
that $\Dress_{\eta_1} \circ \Dress_{\eta_2} = \Dress_{\eta_1 \eta_2}$ and that
\be
\Xi_R(\eta_1 \eta_2 b_1)^{-1} = \Xi_R(\eta_1 \Dress_{\eta_2}(b_1))^{-1} \Xi_R(\eta_2 b_1)^{-1}.
\label{C6}\ee
The last equality is verified by substituting
 \be
\Xi_R(\eta_1 \eta_2 b_1)^{-1}= (\Dress_{\eta_1 \eta_2}(b_1))^{-1} \eta_1 \eta_2 b_2,
\ee
and similarly rewriting the two factors on the right side of \eqref{C6}.
Having verified that \eqref{actapp} is a $G$-action,
it  remains to ascertain that
\be
X_{B\times B}(b_1,b_2) = \dt \cA_{e^{tX}} (b_1,b_2).
\ee
The first component of this equality is obvious from \eqref{dress}, and second one is seen from
\be
\dt \Xi_R(e^{tX} b_1)^{-1} = \dt \left( \left(\Dress_{e^{tX}} (b_1)\right)^{-1} e^{tX} b_1\right)
= -(b_1^{-1} X b_1)_\cB + b_1^{-1}X b_1 = (b_1^{-1} X b_1)_\cG.
\ee

The final step is to convert the action \eqref{actapp} on $B\times B$ into the action
$\hat \cA_\eta$ \eqref{actonPP} on $\fP_- \times \fP$ by means of the map $\mu$ \eqref{mu}.
The desired result, $\hat \cA_\eta \circ \mu = \mu \circ \cA_\eta$,
follows immediately from the identity
\be
(\Dress_\eta (b_1))^{-1} = \Dress_{\Xi_R(\eta b_1)^{-1}} (b_1^{-1}),
\label{appid}\ee
because $\nu$ \eqref{nu} intertwines the dressing action \eqref{Dress} on $B$ and the conjugation action on $\fP$.
The identity \eqref{appid} itself is obtained by applying $\Lambda_L$ \eqref{XiLaT} to both sides of the equality
\be
(\Dress_\eta (b_1))^{-1} = \Xi_R(\eta b_1)^{-1} b_1^{-1} \eta^{-1}.
\ee
The moment maps $\hat \Lambda$ \eqref{momPP} and $J$ above are related by $\hat \Lambda \circ \mu = J$, and thus we have
indeed established the form of the Poisson action on $\fP_- \times \fP$ generated by $\hat \Lambda$.

Incidentally, the formula \eqref{actM} of the quasi-adjoint action can be verified following a
train of thoughts similar to the one presented in this appendix.

\section{Poisson reduction via the shifting trick}
\label{sec:B}

We now explain the origin of the defining equation \eqref{zeta2} of the `spin variable' $\lambda$ by
utilizing the so-called shifting trick of Hamiltonian reduction \cite{OR}.
In the context of Marsden--Weinstein type reductions, the shifting trick means that one first extends the phase space
by a coadjoint orbit or dressing orbit, and then
reduces the extended phase space at the trivial moment map value.
Under mild conditions, the outcome is equivalent to the result of the corresponding
`point reductions' based on taking a moment map value from the `opposite' orbit.

In our case, we may start with the extended Heisenberg double
\be
M_\ext := M \times B = \{ (K,S)\mid K \in M,\, S\in B\},
\label{Mext}\ee
and equip it with the direct product Poisson structure $\br{-,-}_\ext$ built from $\br{-,-}_+$ \eqref{A1T} on $M$ and
$\br{-,-}_B$ \eqref{PBBT} on $B$.
This extended phase space carries the extended moment map $\Lambda_\ext: M_\ext \to B$,
\be
\Lambda_\ext(K,S): = \Lambda(K) S = \Lambda_L(K) \Lambda_R(K) S,
\ee
which generates a Poisson action of $(G,\br{-,-}_G)$ \eqref{PBGT} on $M_\ext$.
Then, we reduce the extended phase space by imposing the moment map constraint
\be
\Lambda_\ext(K,S) = e.
\label{ext1B}\ee
By using that on the `constraint surface' $S = \Lambda(K)^{-1}$, one arrives at the identification
\be
\Lambda_\ext^{-1}(e)/ G \simeq M/G.
\label{shiftid1}\ee
Moreover, for every $G$-invariant function on $\Lambda_\ext^{-1}(e)$ one can define a $G$-invariant
function on $M_\ext$ in such a way that the extended function does not depend on $S$.
In this way, one obtains the identification
\be
C^\infty(\Lambda_\ext^{-1}(e))^G \simeq C^\infty(M)^G.
\label{shiftid2}\ee
Essentially because \eqref{ext1B} represents first class constraints in Dirac's sense \cite{HT},
the identification \eqref{shiftid2} holds at the level of reduced Poisson algebras as well.

On the other hand, coming to the crux, we may introduce a convenient partial gauge fixing
in the moment map constraint surface \eqref{ext1B} by imposing the condition that
$\Xi_R(K) \in G_0$.  Then $K\in M$ can be presented as
\be
K = (Q^{-1} b^{-1} Q) Q^{-1} = Q^{-1} b^{-1} \quad \hbox{with}\quad Q=\Xi_R(K)\in G_0,\,\, b= \Lambda_R(K) \in B.
\label{Kfixapp}\ee
On this  `gauge slice', applying the parametrization
$b = e^p b_+$ (with $p\in \cB_0, b_+ \in B_+$), we get
\be
\Lambda(K) = Q^{-1} b^{-1} Q b = Q^{-1} b_+^{-1} Q b_+.
\ee
Consequently, the moment map condition \eqref{ext1B} becomes
\be
b_+^{-1} Q^{-1} b_+ Q = S.
\label{Seq}\ee
This relation enforces that $S\in B_+$, and after re-naming $S$ to $\lambda$ we recognize the formula \eqref{Seq}
as equation \eqref{zeta2} that we started with in Section \ref{ss:61}.
By imposing the additional condition $Q\in G^\reg$, one may ensure that the residual gauge transformations
of the partial gauge fixing \eqref{Kfixapp} are associated with the normalizer $\fN$ of $G_0$
(which means
that $\eta$ in \eqref{actM}
is restricted so that $\Xi_R (\eta \Lambda_L(K))$ belongs to $\fN$).

The shifting trick was applied in \cite{F1} working in the symplectic framework,
by restricting the variable $S$ in \eqref{Mext} to a dressing orbit of $G$ in $B$ throughout the procedure.

\section{Derivation of the trigonometric RS model}
\label{sec:C}

It is known \cite{FK0} that the standard (real) trigonometric Ruijsenaars--Schneider (RS) model \cite{RS} can be derived
by a specific Marsden--Weinstein type  reduction of the Heisenberg double of
the unitary group $\mathrm{U}(n)$.
The goal of this appendix is to explain how this result can be recovered in the framework of the present paper.
Here, we take $G:= \mathrm{SU}(n)$, and obtain the model in the `center of mass frame'.
In this case the group $B$ consists of the upper triangular matrices in $G^\bC = \mathrm{SL}(n,\bC)$
having positive diagonal entries.  The diagonal elements of the matrices in $B_+<B$ are equal to $1$.
The crucial point is that we restrict the variable $\lambda$ \eqref{zeta2} to a minimal dressing orbit,
of dimension $2(n-1)$, which leads to a symplectic leaf in the reduced phase space.
There exists a one parameter family of such orbits, and their parameter will appear as the coupling
constant of the RS model.

The minimal dressing orbits at issue admit representatives of the form
\be
\Delta(x):= \exp\left( \diag((n-1)x/2, -x/2,\cdots, -x/2 )\right), \quad\hbox{for}\quad x \in \bR^*,
\ee
where the eigenvalue $e^{-x/2}$ of $\Delta(x)$ has  multiplicity $(n-1)$.
Let $\cO(x)$ denote the dressing orbit through $\Delta(x)$.
The only redundancy of these representatives occurs for $n=2$, in which case $\Delta(x)$ and $\Delta(-x)$ lie on the same orbit.
Another representative of the orbit $\cO(x)$  is the upper-triangular matrix $\nu(x) \in B_+$, whose diagonal entries are equal to $1$ and
\be
\nu(x)_{jk} = (1- e^{-x}) \exp((k-j)x/2), \qquad \forall j<k.
\ee
This matrix satisfies the equality
\be
\nu(-x) = \nu(x)^{-1}.
\label{nu-x}\ee
More importantly,  as was shown in \cite{FK0}, one has
\be
\cO(x) \cap B_+ = \{ T \nu(x) T^{-1} \mid T \in G_0\},
\label{restorb}\ee
where $G_0$ is the standard maximal torus of $G=\mathrm{SU}(n)$.  Thus,  $(\cO(x) \cap B_+)/G_0$ consists of a single point.

The defining equation \eqref{zeta2} entails that $\lambda$ belongs to the subgroup $B_+ < B$, and we know from
 Theorem \ref{thm:5.1}  that $\lambda$ can be restricted to the intersection of $B_+$ with any dressing orbit.
Now we choose to restrict it to the orbit $\cO(-x)$ with a fixed $x\in \bR^*$.
On account of the relations \eqref{nu-x} and \eqref{restorb}, we then obtain a complete  fixing of the residual $G_0$ `gauge freedom'
by imposing the condition
\be
\lambda = b_+^{-1}  Q^{-1} b_+ Q = \nu(x)^{-1}.
\label{numom}\ee

The paper \cite{FK0} analysed the symplectic reduction of the Heisenberg double based on the moment
map constraint
\be
\Lambda(K) = \nu(x)
\label{momLMP}\ee
with $\Lambda$ in \eqref{Lambda}.
After diagonalizing  $g_R = \Xi_R(K)$ \eqref{KdecT},  i.e., by setting
 $g_R = Q = \diag(Q_1,\dots, Q_n) \in G_0$,
$K\in \mathrm{SL}(n,\bC)$ takes the form
\be
K  = Q^{-1} b^{-1},
\quad \hbox{with some}\quad b = e^p b_+,
\label{KparappC}\ee
where $p\in \cB_0$ and $b_+ \in B_+$.  Then the  constraint \eqref{momLMP} leads precisely to equation \eqref{numom}.
It was proved in \cite{FK0} that \eqref{numom}  implies that $Q\in G_0^\reg$ and
 $b_+$ in \eqref{numom} can be expressed  in terms of $Q\in G_0^\reg$  as follows:
\be
(b_+)_{kl} =Q_k \bar Q_l\prod_{m=1}^{l-k} \frac{ e^{-\frac{x}{2}} \bar Q_k - e^{\frac{x}{2}} \bar Q_{k+m-1}}{\bar Q_k - \bar Q_{k+m}},
\quad  1\leq k < l \leq n,
\label{b+Q}\ee
where $\bar Q_k = Q_k^{-1}$. Of course, $(b_+)_{kk}=1$ and the matrix elements below the diagonal are zero.
(The correspondence between our notations and those in \cite{FK0} is explained in the subsequent Remark \ref{rem:D1}.)

We have taken the quotient by the $G_0$-symmetry, but there still remains a residual $S_n = \fN/G_0$ redundancy in our description.
Consequently, the variables $Q$ and $p$ parametrize an $S_n$ covering space of a Poisson  subspace of the  full
reduced phase space.
Since $\lambda$ became a constant by the gauge fixing,
it follows from  \eqref{FHred-form} that the Poisson bracket on this covering space
is given by the formula
\be
\br{F,H}^\red(Q,p) = \langle D_Q F, d_p H \rangle_\bI - \langle D_Q H, d_p F \rangle_\bI,
\ee
which corresponds to the symplectic form
\be
\omega_\red = \Im\tr(dp \wedge Q^{-1} dQ).
\ee
The elements of $S_n$ permute the $n$ diagonal  entries of the matrix $Q$.
However, a careful analysis \cite{FK1} shows that their action on $p=\diag(p_1,\dots, p_n)$  has a complicated form,
and what are permuted in the obvious manner are the entries of the traceless diagonal matrix $\vartheta$ given by
the following formula:
\be
\vartheta_k = 2 p_k -
\frac{1}{2} \sum_{m<k} \ln\left[1+  \frac{\sinh^2(x/2)}{\sin^2(q_k - q_m)}\right] +
\frac{1}{2} \sum_{m>k} \ln\left[1+  \frac{\sinh^2(x/2)}{\sin^2(q_k - q_m)}\right],
\quad k=1,\dots, n.
\label{thetak}\ee
Here, we use the parametrization
\be
Q = \exp( 2\ri q) \quad \hbox{with} \quad q = \diag(q_1,\dots, q_n),\quad \tr(q)=0.
\ee
Equation \eqref{thetak} yields a canonical transformation, since in terms of $q$ and $\vartheta$ one has
\be
\omega_\red = \tr (d \vartheta \wedge dq).
\ee
This means that $Q$ and $\vartheta$ are the natural variables on $T^* G_0^\reg$.

We are also reducing the `free Hamiltonians' given by the dressing invariant functions of  $b_R = b = e^p b_+$.
The main Hamiltonian of the reduced system is
\be
H_{\mathrm{RS}} = \frac{1}{2}( H_{+, \mathrm{RS}} + H_{-,\mathrm{RS}}), \quad \hbox{with}\quad H_{+, \mathrm{RS}} = \tr(bb^\dagger), \quad H_{-,\mathrm{RS}} = \tr(bb^\dagger)^{-1}.
\ee
By using \eqref{b+Q} and the canonical transformation \eqref{thetak} to express $p$ in terms of $q$ and $\vartheta$,  one finds
\be
H_{\pm, \mathrm{RS}}(q,\vartheta) =
\sum_{k=1}^n e^{\pm \vartheta_k} \prod_{m\neq k}\left[1+  \frac{\sinh^2(x/2)}{\sin^2(q_k - q_m)}\right]^{\frac{1}{2}},
\ee
and $H_{\mathrm{RS}}(q,\vartheta)$ is just the standard trigonometric Ruijsenaars--Schneider
Hamiltonian introduced in \cite{RS}.
In conclusion, we have shown that the reduction of the Heisenberg double gives
the trigonometric RS system on a symplectic leaf of the reduced phase space, which is
symplectomorphic to $(T^* G_0^\reg)/S_n$.
It is worth noting  that this system is Liouville integrable, but is not superintegrable \cite{RIMS95}.

\begin{rem}\label{rem:D1}
It follows from the results of \cite{FK0}  that the isotropy group of the elements
$\Lambda^{-1}(\cO(x))$ is the center of $G=\mathrm{SU}(n)$, i.e., $\Lambda^{-1}(\cO(x))$ is a subset of $M_*$ \eqref{M*}.
In that paper the variables $b_L$ and $g_R$ constituting $K= b_L g_R^{-1}$ \eqref{KdecT}
were used, while here we mostly worked with $b\equiv b_R$ and $g\equiv g_R$.
After bringing $g_R$ into $G_0$, the relation of the variables becomes
$b_L = Q^{-1} b^{-1} Q$.
In \cite{FK0} $b_L$ was parametrized as $b_L=\cN a$ with $\cN \in B_+$ and $a\in B_0$.
These are related to our variables $b_+$ and $p$ by $\cN^{-1} = Q^{-1} b_+ Q$ and $a = e^{-p}$.
Finally, for readers  interested in a  detailed comparison with \cite{FK0,FK1}, we also note that
the notations $T_k$ and $\zeta_k = \ln a_k$ in \cite{FK0}
correspond to $Q_k$ and $-p_k$ as used in the present paper;
and what is denoted by $p_k$ in \cite{FK0} corresponds to $\vartheta_k$  \eqref{thetak}.
\end{rem}


\begin{thebibliography}{99}

 \setlength{\parskip}{0.12em}


\bibitem{Al}
A.Yu.~Alekseev,
{\it On Poisson actions of compact Lie groups on symplectic manifolds.}
J. Diff. Geom. {\bf 45} (1997) 241-256;
\href{https://arxiv.org/abs/dg-ga/9602001}{\tt arXiv:dg-ga/9602001}

 \bibitem{AM}
 A.Yu.~Alekseev and  A.Z.~Malkin,
 {\it Symplectic structures associated to Lie--Poisson groups}.
   Commun. Math. Phys. {\bf 162} (1994) 147-173;
 \href{https://arxiv.org/abs/hep-th/9303038}{\tt arXiv:hep-th/9303038}


 \bibitem{AKSM}
A.~Alekseev, Y.~Kosmann-Schwarzbach and E.~Meinrenken,
{\it Quasi-Poisson manifolds}.
Canad. J. Math. {\bf 54} (2002) 3-29;
\href{https://arxiv.org/abs/math/0006168}{\tt arXiv:math/0006168}

\bibitem{AMM}
A.~Alekseev, A.~Malkin and E.~Meinrenken,
{\it Lie group valued moment maps}.
J. Differential Geom. {\bf 48} (1998) 445-495;
 \href{https://arxiv.org/abs/dg-ga/9707021}{\tt arXiv:dg-ga/9707021}

\bibitem{ARe}
S.~Arthamonov and N.~Reshetikhin,
{\it Superintegrable systems on moduli spaces of flat connections}.
Commun. Math. Phys. {\bf 386} (2021) 1337-1381;
 \href{https://arxiv.org/abs/1909.08682}{\tt arXiv:1909.08682}

\bibitem{A}
 G.~Arutyunov,
Elements of Classical and Quantum Integrable Systems, Springer, 2019

\bibitem{AFM}
G.E.~Arutyunov, S.A.~Frolov and P.B.~Medvedev,
{\it Elliptic Ruijsenaars-Schneider model from the cotangent bundle over the two-dimensional current group.}
J. Math. Phys.  {\bf 38} (1997) 5682-5689;
\href{https://arxiv.org/abs/hep-th/9608013}{\tt arXiv:hep-th/9608013}


\bibitem{AO}
G.~Arutyunov and E.~Olivucci,
{\it Hyperbolic spin Ruijsenaars--Schneider model from Poisson reduction}.
Proc. Steklov Inst. Math. {\bf 309} (2020) 31-45;
 \href{https://arxiv.org/abs/1906.02619}{\tt  arXiv:1906.02619}

\bibitem{CF}
O.~Chalykh and M.~Fairon,
{\it On the Hamiltonian formulation of the trigonometric spin Ruijsenaars--Schneider system}.
 Lett. Math. Phys. {\bf 110} (2020) 2893-2940;
 \href{https://arxiv.org/abs/1811.08727}{\tt arXiv:1811.08727}


\bibitem{CP}
V.~Chari and  A.N.~Pressley,
A Guide to Quantum Groups,
Cambridge University Press, 1995

\bibitem{Drin}
V.G.~Drinfel'd,
{\it Hamiltonian structures on Lie groups, Lie bialgebras and the geometric meaning of the classical
Yang--Baxter equations}.
Soviet Math. Dokl. {\bf 27} (1983) 68-71


\bibitem{DK}
J.J.~Duistermaat and J.A.C.~Kolk,
Lie Groups, Springer, 2000

\bibitem{Eti}
P.~Etingof,
Calogero--Moser Systems and Representation Theory,
European Mathematical Society, 2007

\bibitem{EFK}
P.I.~Etingof, I.B.~Frenkel and A.A.~Kirillov Jr.,
{\it Spherical functions on affine Lie groups}.
Duke Math. J. {\bf 80} (1995) 59-90;
\href{https://arxiv.org/abs/hep-th/9407047}{\tt arXiv:hep-th/9407047}

 \bibitem{EV}
 P. Etingof and A. Varchenko,
 {\it  Geometry and classification of solutions of the classical dynamical Yang--Baxter equation}.
Commun. Math. Phys. {\bf 192} (1998) 77-120;
\href{https://arxiv.org/abs/q-alg/9703040}{\tt arXiv:q-alg/9703040}


\bibitem{FF}
 M.~Fairon and L.~Feh\'er,
{\it Integrable multi-Hamiltonian systems from reduction of an extended quasi-Poisson double of $\mathrm{U}(n)$}.
Ann. Henri Poincar\'e {\bf 24} (2023)  3461-3529;
preptrint \href{https://arxiv.org/abs/2302.14392}{\tt arXiv:2302.14392}


\bibitem{FFM}
 M.~Fairon, L.~Feh\'er and I.~Marshall,
{\it Trigonometric real form of the spin RS model of Krichever and Zabrodin}.
Ann. Henri Poincar\'e {\bf 22} (2021) 615--675;
\href{https://arxiv.org/abs/2007.08388}{\tt arXiv:2007.08388}


\bibitem{Fas}
F.~Fasso,
{\it Superintegrable Hamiltonian systems: geometry and perturbations}.
Acta Appl. Math. {\bf 87} (2005) 93-121

\bibitem{F1}
L.~Feh\'er,
 {\it Poisson--Lie analogues of spin Sutherland models}.
Nucl. Phys. B {\bf 949} (2019) Paper No. 114807;
 \href{https://arxiv.org/abs/1809.01529}{\tt arXiv:1809.01529}


\bibitem{F3}
L.~Feh\'er,
{\it Reduction of a bi-Hamiltonian hierarchy on $T^*\mathrm{U}(n)$ to spin Ruijsenaars--Sutherland models.}
Lett. Math. Phys. {\bf 110} (2020) 1057-1079;
\href{https://arxiv.org/abs/1908.02467}{\tt arXiv:1908.02467}

\bibitem{F2}
L.~Feh\'er,
{\it Poisson reductions of master integrable systems on doubles of compact Lie groups}.
 Ann. Henri Poincar\'e {\bf 24} (2023) 1823-1876;
 \href{https://arxiv.org/abs/2208.03728}{\tt arXiv:2208.03728}


\bibitem{WGMPnew}
L.~Feh\'er,
Notes on the degenerate integrability
 of reduced systems obtained from the master systems of free motion on cotangent bundles
 of compact Lie groups, preprint \href{https://arxiv.org/abs/2309.16245}{\tt arXiv:2309.16245}

\bibitem{FK0}
L.~Feh\'er and C.~Klim\v c\'\i k,
{\it Poisson--Lie generalization of the Kazhdan-Kostant-Sternberg reduction}.
Lett. Math. Phys. {\bf 87} (2009) 125-138;
\href{https://arxiv.org/abs/0809.1509}{\tt arXiv:0809.1509}


\bibitem{FK1}
L.~Feh\'er and C.~Klim\v c\'\i k,
{\it Poisson--Lie interpretation of trigonometric Ruijsenaars duality}.
Commun. Math. Phys. {\bf 301} (2011) 55-104;
  \href{https://arxiv.org/abs/0906.4198}{\tt arXiv:0906.4198}


\bibitem{FK2}
 L.~Feh\'er and C.~Klim\v c\'ik,
 {\it Self-duality of the compactified Ruijsenaars--Schneider system from quasi-Hamiltonian reduction}.
  Nucl. Phys. B {\bf 860} (2012) 464-515;
  \href{https://arxiv.org/abs/1101.1759}{\tt arXiv:1101.1759}


 \bibitem{FM}
L.~Feh\'er and I.~Marshall,
 {\it On the bi-Hamiltonian structure of the trigonometric spin Ruijsenaars--Sutherland hierarchy.}
   Geometric Methods in Physics XXXVIII, eds. P. Kielanowski et al, Birkhauser,  pp. 75-87, 2020;
   \href{https://arxiv.org/abs/2007.09658}{\tt arXiv:2007.09658}

\bibitem{FP1}
L.~Feh\'er and B.G.~Pusztai,
{\it A class of Calogero type reductions of free motion on a simple Lie group}.
 Lett. Math. Phys. {\bf 79} (2007) 263-277;
 \href{https://arxiv.org/abs/math-ph/0609085}{\tt arXiv:math-ph/0609085}

 \bibitem{FP2}
L.~Feh\'er and B.G.~Pusztai,
{\it Twisted spin Sutherland models from quantum Hamiltonian reduction.}
 J. Phys. A: Math. Theor. {\bf 41} (2008) Paper No. 194009;
 \href{https://arxiv.org/abs/0711.4015}{\tt arXiv:0711.4015}



\bibitem{FR}
V.V.~Fock and A.A.~Rosly,
{\it Poisson structure on moduli of flat connections on Riemann surfaces and the $r$-matrix.}
Moscow Seminar in Mathematical Physics,
AMS Transl. Ser. 2, Vol.~191, pp.~67-86, 1999;
\href{https://arxiv.org/abs/math/9802054}{\tt arXiv:math/9802054}


\bibitem{Fris}
I.~Fris, V.~Mandrosov, Ya.A.~Smorodinsky, M.~Uhlir M and P.~Winternitz,
{\it On higher symmetries in quantum mechanics.}
 Phys. Lett. {\bf 16} (1965) 354–356

\bibitem{GN}
A.~Gorsky and N.~Nekrasov,
{\it Relativistic Calogero-Moser model as gauged WZW theory.}
Nucl. Phys. B {\bf 436} (1995) 582-608;
\href{https://arxiv.org/abs/hep-th/9401017}{\tt arXiv:hep-th/9401017}


\bibitem{HT}
M. Henneaux and C. Teitelboim, Quantization of Gauge Systems, Princeton University
Press, 1992

\bibitem{J}
B.~Jovanovic,
{\it Symmetries and integrability}.
Publ. Institut Math. {\bf 49} (2008) 1-36;
 \href{https://arxiv.org/abs/0812.4398}{\tt arXiv:0812.4398}

\bibitem{KKS}
D.~Kazhdan, B.~Kostant and S.~Sternberg,
{\it Hamiltonian group actions and dynamical systems of Calogero type}.
Comm. Pure Appl. Math. {\bf 31} (1978) 481-507

\bibitem{KLOZ}
 S.~Kharchev, A.~Levin, M.~Olshanetsky and  A.~Zotov,
 {\it Quasi-compact Higgs bundles and Calogero--Sutherland systems with two types spins}.
 J. Math. Phys. {\bf 59} (2018) Paper No. 103509;
  \href{https://arxiv.org/abs/1712.08851}{\tt arXiv:1712.08851}

\bibitem{Kli}
C. Klim\v c\'\i k,
{\it On moment maps associated to a twisted Heisenberg double}.
 Rev. Math. Phys. {\bf 18} (2006) 781-821;
\href{https://arxiv.org/abs/math-ph/0602048}{\tt  arXiv:math-ph/0602048}


\bibitem{Knapp}
A.W.~Knapp, Lie Groups Beyond an Introduction, Birkh\"auser, 1996


\bibitem{LMV}
C.~Laurent-Gengoux, E.~Miranda and P.~Vanhaecke,
{\it Action-angle coordinates for integrable systems on Poisson manifolds}.
Int. Math. Res. Not. {\bf  2011} 1839-1869;  	
\href{https://arxiv.org/abs/0805.1679}{\tt arXiv:0805.1679 [math.SG]}

\bibitem{LuPhD}
J-H.~Lu,
{\it Multiplicative and Affine Poisson Structures on Lie Groups},
\href{http://hkumath.hku.hk/~jhlu/thesis.pdf}{PhD thesis}, Berkeley, 1990


\bibitem{Lu}
J-H.~Lu,
{\it Momentum mappings and reduction of Poisson actions}.
 Symplectic Geometry, Groupoids, and Integrable Systems,
 Springer, pp.~209-226, 1991

 \bibitem{M}
I.~Marshall,
\textit{A new model in the Calogero-Ruijsenaars family}.
Commun. Math. Phys. {\bf 338} (2015) 563-587;
\href{https://arxiv.org/abs/1311.4641}{\tt arXiv:1311.4641}

\bibitem{Mein}
E.~Meinrenken,
{\it Verlinde formulas for nonsimply conected groups}.
  Lie Groups, Geometry, and Representation Theory, Progress in Mathematics (Birkh\"auser) {\bf 326} (2018) 381-417;
\href{https://arxiv.org/abs/1706.04045}{\tt arXiv:1706.04045}

 \bibitem{Mic}
P.W.~Michor, Topics in Differential Geometry, American Mathematical Society, 2008

 \bibitem{MF}
A.S.~Mischenko and A.T.~Fomenko,
{\it  Generalized Liouville method for integrating Hamiltonian systems}.
Funct. Anal. Appl. {\bf 12} (1978) 113-125

 \bibitem{MPV}
 W.~Miller Jr, S.~Post and P.~Winternitz,
 {\it  Classical and quantum superintegrability with applications}.
 J. Phys. A {\bf 46} (2013) Paper No. 423001;
 \href{https://arxiv.org/abs/1309.2694}{\tt  arXiv:1309.2694}

  \bibitem{Nekh}
N.N.~Nekhoroshev,
{\it  Action-angle variables and their generalizations}.
Trans. Moscow Math. Soc. {\bf 26} (1972) 180-197

\bibitem{N}
N.~Nekrasov,
{\it Infinite-dimensional algebras, many-body systems and gauge theories.}
 Moscow Seminar in Mathematical Physics,
AMS Transl. Ser. 2, Vol.~191, Amer. Math. Soc., pp. 263-299, 1999

\bibitem{Ob}
A.~Oblomkov,
{\it Double affine Hecke algebras and Calogero-Moser spaces.}
Represent. Theory {\bf 8} (2004)  243-266;
\href{https://arxiv.org/abs/math/0303190}{\tt arXiv:math/0303190}


\bibitem{OP}
M.A. Olshanetsky and A.M. Perelomov,
{\it Classical integrable finite-dimensional systems related to Lie algebras.}
Phys. Rept. {\bf 71} (1981) 313-400

\bibitem{OR}
J.-P.~Ortega and T.~Ratiu,
Momentum Maps and Hamiltonian Reduction, Birkh\"auser, 2004

\bibitem{Per}
A.M. Perelomov,
Integrable Systems of Classical Mechanics and Lie Algebras,
Birkh\"auser, 1990

\bibitem{PolR}
A.P.~Polychronakos,
{\it Physics and mathematics of Calogero particles}.
J. Phys. A: Math. Gen. {\bf 39} (2006) 12793-12845;
\href{https://arxiv.org/abs/hep-th/0607033}
{\tt arXiv:hep-th/0607033}


\bibitem{Re1}
N.~Reshetikhin,
{\it Degenerate integrability of spin Calogero--Moser systems and the
duality with the spin Ruijsenaars systems}.
Lett. Math. Phys. {\bf 63} (2003) 55-71;
\href{https://arxiv.org/abs/math/0202245}{\tt arXiv:math/0202245}


\bibitem{Re2}
 N.~Reshetikhin,
{\it Degenerately integrable systems}.
J. Math. Sci. {\bf 213} (2016) 769-785;
\href{https://arxiv.org/abs/1509.00730}{\tt arXiv:1509.00730}

\bibitem{Re3}
N.~Reshetikhin,
{\it Spin Calogero--Moser models on symmetric spaces}.
Integrability, Quantization, and Geometry. I. Integrable Systems, Proc. Sympos. Pure Math., Vol. 103.1, Amer. Math. Soc.,
 pp. 377-402, 2021;
\href{https://arxiv.org/abs/1903.03685}{\tt arXiv:1903.03685}

\bibitem{Re03}
N.~Reshetikhin,
{\it Periodic and open classical spin Calogero--Moser chains}.
preprint \href{https://arxiv.org/abs/2302.14281}{\tt arXiv:2302.14281}

\bibitem{Rud}
G.~Rudolph and M.~Schmidt,
Differential Geometry and Mathematical Physics. Part I. Manifolds, Lie Groups and
Hamiltonian Systems, Springer, 2013

\bibitem{RIMS95}
S.N.M.~Ruijsenaars,
{\it Action-angle maps and scattering theory for some finite-dimensional
integrable systems III. Sutherland type systems and their duals}.
Publ. RIMS {\bf 31} (1995) 247-353

\bibitem{RS}
S.N.M. Ruijsenaars and H. Schneider,
{\it A new class of integrable systems and its relation to solitons}.
Ann. Phys. {\bf 170} (1986) 370-405


\bibitem{Sam}
H.~Samelson, Notes on Lie Algebras, third corrected edition, Springer, 1990


\bibitem{SL}
R. Sjamaar and E. Lerman,
{\it Stratified symplectic spaces and reduction}.
Ann. of Math.  {\bf 134} (1991) 375-422

\bibitem{STS}
M.A.~Semenov-Tian-Shansky,
{\it Dressing transformations and Poisson group actions}.
Publ. RIMS {\bf 21} (1985)  1237-1260


\bibitem{STSrev}
M.A.~Semenov-Tian-Shansky,
{\it Integrable systems: an r-matrix approach}.
Kyoto preprint RIMS-1650, 2008;
\href{http://www.kurims.kyoto-u.ac.jp/preprint/file/RIMS1650.pdf}{kurims.kyoto-u.ac.jp/preprint/file/RIMS1650.pdf}


\bibitem{Wo}
S.~Wojciechowski,
{\it An integrable marriage of the Euler equations with the Calogero--Moser system.}
Phys. Lett. A \textbf{111}  (1985) 101-103


\end{thebibliography}
\end{document}